\documentclass{amsart}
\usepackage{graphics}
\usepackage{graphicx}
\usepackage[shortlabels]{enumitem}
\usepackage{amsfonts,amssymb,amsopn,amscd,amsmath}
\usepackage[usenames,dvipsnames]{xcolor}
\usepackage{tikz,xspace,titletoc}
\usepackage[titletoc]{appendix}
\usepackage{algpseudocode}
\usetikzlibrary{matrix,arrows,calc,cd,decorations.pathmorphing}
\usepackage{mathrsfs}
\usepackage{multirow}

\usepackage[utf8]{inputenc}
\allowdisplaybreaks

\usepackage[usenames,dvipsnames]{xcolor}
\definecolor{darkblue}{rgb}{0.0, 0.0, 0.55}
\definecolor{bordeaux}{rgb}{0.34, 0.01, 0.1}

\usepackage[pagebackref,colorlinks,linkcolor=bordeaux,citecolor=darkblue,urlcolor=black,hypertexnames=true]{hyperref}

 \newtheorem{theorem}{Theorem}[section]
 \newtheorem{corollary}[theorem]{Corollary}
 \newtheorem{lemma}[theorem]{Lemma}{\rm}
 \newtheorem{proposition}[theorem]{Proposition}

 \newtheorem{definition}[theorem]{Definition}{\rm}
 
 \newtheorem{remark}[theorem]{Remark}
 \newtheorem{example}[theorem]{Example}

\newtheorem{algorithm}[theorem]{Algorithm}

\numberwithin{equation}{section}

\DeclareMathOperator{\Sym}{Sym}
\DeclareMathOperator{\Diag}{Diag}

\DeclareMathOperator{\I}{I}
\DeclareMathOperator{\II}{II}
\DeclareMathOperator{\Trace}{tr}

\DeclareMathOperator{\rank}{rank}

\DeclareMathOperator{\id}{id}
\newif\ifcomment
\commentfalse
\commenttrue

\newtheorem{thmA}{Theorem}

\newcommand{\revise}[1]{{{\color{black}#1}}}
\begin{document}
\def\cA{\mathcal A}
\def\cH{\mathcal H}
\def\red{\color{red}}
\def\bl{\color{blue}}
\def\ora{\color{orange}}
\def\green{\color{green}}
\def\br{\color{brown}}
\newcommand\vna{von Neumann algebra\xspace}
\newcommand\vnas{von Neumann algebras\xspace}
\newcommand{\realtofloat}{\mathtt{Real2Float}}
\newcommand{\sparsepop}{\mathtt{SparsePOP}}
\newcommand{\ncsostools}{\mathtt{NCSOStools}}
\newcommand{\ncpoltosdpa}{\mathtt{Ncpol2sdpa}}
\newcommand{\gloptipoly}{\mathtt{Gloptipoly}}
\newcommand{\ax}{\langle \underline{X} \rangle}
\newcommand{\trvN}{\text{vN}}
\newcommand{\twoone}{\II_1}
\def\la{\langle}
\def\ra{\rangle}
\def\e{{\rm e}}
\def\x{\mathbf{x}}
\def\by{\mathbf{y}}
\def\bz{\mathbf{z}}
\def\cC{\mathcal{C}}
\def\R{\mathbb{R}}
\def\Mbb{\mathbb{M}}
\def\Sbb{\mathbb{S}}
\def\T{\mathbb{T}}
\def\N{\mathbb{N}}
\def\K{\mathbb{K}}
\def\bK{\overline{\mathbf{K}}}
\def\Q{\mathbf{Q}}
\def\M{\mathbf{M}}
\def\O{\mathbf{O}}
\def\C{\mathbb{C}}
\def\Hbb{\mathbb{H}}
\def\P{\mathbf{P}}
\def\Z{\mathbb{Z}}
\def\A{\mathbf{A}}
\def\W{\mathbf{W}}
\def\bfone{\mathbf{1}}
\def\V{\mathbf{V}}
\def\AA{\overline{\mathbf{A}}}
\def\c{\mathbf{C}}
\def\bL{\mathbf{L}}
\def\bS{\mathbf{S}}
\def\Y{\mathbf{Y}}
\def\X{\mathbf{X}}
\def\G{\mathbf{G}}
\def\Bbb{\mathbb{B}}
\def\Dbb{\mathbb{D}}
\def\f{\mathbf{f}}
\def\z{\mathbf{z}}
\def\y{\mathbf{y}}
\def\d{\hat{d}}
\def\bx{\mathbf{x}}
\def\y{\mathbf{y}}
\def\h{\mathbf{h}}
\def\u{\mathbf{u}}
\def\g{\mathbf{g}}
\def\w{\mathbf{w}}
\def\cX{\mathcal{X}}
\def\a{\mathbf{a}}
\def\q{\mathbf{q}}
\def\u{\mathbf{u}}
\def\vb{\mathbf{v}}
\def\s{\mathcal{S}}
\def\cD{\mathcal{D}}
\def\co{{\rm co}\,}
\def\cp{{\rm CP}}
\def\fatT{\T}
\def\skinnyT{\text T}
\def\tg{\tilde{f}}
\def\tx{\tilde{\x}}
\def\supmu{{\rm supp}\,\mu}
\def\supnu{{\rm supp}\,\nu}
\def\m{\mathcal{M}}
\def\bR{\mathbf{R}}
\def\om{\mathbf{\Omega}}
\def\c{\mathbf{c}}
\def\s{\mathcal{S}}
\def\k{\mathcal{K}}
\def\la{\langle}
\def\ra{\rangle}
\def\opt{\text{opt}}
\def\cyc{\overset{\text{cyc}}{\sim}}
\def\starcyc{\underset{\vspace{-10mm}\star}{\overset{\text{cyc}}{\sim}}}
\def\starcyc{\begin{smallmatrix}\text{cyc}\\ \sim\\[-.4mm] \star\end{smallmatrix}}

\def\cM{\mathcal{M}}
\def\cQ{{\mathcal M^{\rm cyc}}}
\def\cN{\mathcal{N}}
\def\cF{\mathcal{F}}
\def\cE{\mathcal{E}}
\def\cB{\mathcal{B}}

\def\smileL{\overset{\smallsmile}{L}}
\def\blambda{{\boldsymbol{\lambda}}}
\def\bsigma{{\boldsymbol{\sigma}}}
\def\RX{\R \langle \underline{x} \rangle}
\def\CX{\C \langle \underline{x} \rangle}
\def\TX{\fatT}
\def\KX{\K \langle \underline{x} \rangle}
\def\uX{\underline X}
\def\uA{\underline A}
\def\ux{\underline x}
\def\mx{\langle\underline x\rangle}
\def\uY{\underline Y}
\def\uy{\underline y}
\newcommand{\RXI}[1]{\R \langle \underline{x}(I_{#1}) \rangle }
\def\SigmaX{\Sigma \langle \underline{x} \rangle}
\newcommand{\SigmaXI}[1]{\Sigma \langle \underline{x}(I_{#1}) \rangle }
\def\SymRX{\Sym \RX}
\def\SymCX{\Sym \CX}
\def\SymTX{\Sym \TX}
\def\SymT{\Sym \T}
\def\SymKX{\Sym \KX}
\def\ov{\overline{o}}
\def\und{\underline{o}}
\newcommand{\gns}{\texttt{PureTraceGNS}}
\newcommand{\victor}[1]{\Vi{#1}}
\newcommand{\victorshort}[1]{\todo[inline,color=purple!30]{VM: #1}}
\newcommand{\igorshort}[1]{\todo[color=brown!30]{IK: #1}}
\newcommand{\janez}[1]{{\color{red} Janez: #1}}
\newcommand{\janezshort}[1]{\todo[color=red!30]{TdW: #1}}

\makeatletter
\newcommand{\mycontentsbox}{%
\printindex
{\centerline{NOT FOR PUBLICATION}
\addtolength{\parskip}{-2.0pt}\normalsize
\tableofcontents}}
\def\enddoc@text{\ifx\@empty\@translators \else\@settranslators\fi
\ifx\@empty\addresses \else\@setaddresses\fi
\newpage\mycontentsbox
}
\makeatother

\colorlet{commentcolour}{green!50!black}
\newcommand{\comment}[3]{%
\ifcomment%
	{\color{#1}\bfseries\sffamily(#3)%
	}%
	\marginpar{\textcolor{#1}{\hspace{3em}\bfseries\sffamily #2}}%
	\else%
	\fi%
}
\newcommand{\Ig}[1]{
	\comment{magenta}{I}{#1}
}
\newcommand{\Vi}[1]{
	\comment{blue}{V}{#1}
}
\newcommand{\idea}[1]{\textcolor{red}{#1(?)}}

\newcommand{\Expl}[1]{
	{\tag*{\text{\small{\color{commentcolour}#1}}}%
	}
}

\renewcommand{\algorithmicrequire}{\textbf{Input:}}
\renewcommand{\algorithmicensure}{\textbf{Output:}}
\newcommand{\sparsegns}{\texttt{SparseGNS}}
\newcommand{\sparseeiggns}{\texttt{SparseEigGNS}}
\newcommand{\eigmin}{\texttt{NCeigMin}}
\newcommand{\eigminsparse}{\texttt{NCeigMinSparse}}
\def\nsdp{n_{\text{sdp}}}
\def\msdp{m_{\text{sdp}}}

\title[Optimization over trace polynomials]{Optimization over trace polynomials}

\author{Igor Klep \and Victor Magron \and Jurij Vol\v{c}i\v{c}}

\address{Igor Klep: Faculty of Mathematics and Physics, Department of Mathematics,  University of Ljubljana, Slovenia}
\email{igor.klep@fmf.uni-lj.si}
\thanks{IK was supported by the 
Slovenian Research Agency grants J1-8132 and P1-0222. 
}
\address{Victor Magron: LAAS-CNRS \& Institute of Mathematics from Toulouse, France}
\email{vmagron@laas.fr}
\thanks{VM was supported by the FMJH Program PGMO (EPICS project) and  EDF, Thales, Orange et Criteo, as well as from the Tremplin ERC Stg Grant ANR-18-ERC2-0004-01 (T-COPS project).}
\address{Jurij Vol\v{c}i\v{c}: Department of Mathematics, Texas A\&M University, Texas}
\email{volcic@math.tamu.edu}
\thanks{JV was supported by the NSF grant DMS-1954709.}
\date{}

\begin{abstract}
Motivated by recent progress in quantum information theory,
this article aims at optimizing trace polynomials, i.e., polynomials in noncommuting variables and traces of their products. 
A novel Positivstellensatz certifying positivity of trace polynomials subject to trace constraints is presented, and 
a hierarchy of semidefinite relaxations 
converging monotonically to the optimum of a trace polynomial subject to tracial constraints is provided.
This hierarchy can be seen as a tracial analog of the 
Pironio, Navascu\'es and Ac\'in scheme [New J. Phys., 2008] 
for optimization of noncommutative polynomials.
The Gelfand-Naimark-Segal (GNS) construction is applied to extract optimizers of the trace optimization problem if flatness and extremality conditions are satisfied.
These conditions are sufficient to obtain finite convergence of our hierarchy.
\revise{The results obtained are applied to violations of polynomial Bell inequalities in quantum information theory.}
The main techniques used in this paper are inspired by 
real algebraic geometry, operator theory, and noncommutative algebra.
\end{abstract}



\maketitle

\section{Introduction}
\label{sec:intro}
The goal of this article is to solve the  class of polynomial optimization problems with noncommuting variables (e.g.,~polynomials in matrices) involving the trace.
Applications of interest arise from  quantum theory and quantum information science~\cite{navascues2008convergent,pozas2019bounding,huber2020positive} as well as control theory~\cite{skelton1997unified,engineeringFRAG}. 
Further motivation relates to the generalized Lax conjecture~\cite{lax1957differential}, where the goal is to obtain computer-assisted proofs based on noncommutative sums of squares in Clifford algebras~\cite{netzer2014hyperbolic}.
The verification of noncommutative polynomial trace inequalities has also been motivated by a conjecture formulated by Bessis, Moussa and Villani (BMV) in 1975~\cite{bessis1975monotonic}, which has been recently proved by Stahl \cite{stahl2013proof} (see also the Lieb and Seiringer reformulation~\cite{lieb2004equivalent}). 
Further efforts focused on applications arising from bipartite quantum correlations~\cite{Gribling18}, and matrix factorization ranks in~\cite{Gribling19}.
In a related analytic direction, there has been recent progress on multivariate generalizations of the Golden-Thompson inequality and the Araki-Lieb-Thirring inequality~\cite{GTineq1,GTineq2}.

There is a plethora of prior research in quantum information theory
involving reformulating problems as optimization of noncommutative polynomials.
One famous application is to characterize the set of quantum correlations.
Bell inequalities \cite{bell1964einstein} provide a method to investigate entanglement, which allows two or more parties to be correlated in a non-classical way, and is often studied through the set of bipartite quantum correlations. 
Such correlations consist of the conditional probabilities that two physically separated parties can generate by performing measurements on a shared entangled state.
These conditional probabilities satisfy some  inequalities  classically, but violate them in the quantum realm \cite{clauser1969proposed}.\looseness=-1

Classically, \emph{polynomial optimization} aims at minimizing a polynomial over a set defined by a finite conjunction of polynomial inequalities, i.e.,  a basic closed \emph{semialgebraic set}.
Solving this optimization problem is NP-hard in general~\cite{Laurent:Survey}.
{\em Lasserre's hierarchy}~\cite{Las01sos} is a nowadays well established methodology to handle polynomial optimization in a practical way. 
This framework consists of approximating the solution of the initial problem by considering a hierarchy of convex relaxations.  
Each step of the hierarchy boils down to computing the optimal value of a \emph{semidefinite program} \cite{anjos2011handbook}, that is, the optimum of a linear function under linear matrix inequality constraints.
As a consequence of Putinar's Positivstellensatz \cite{Putinar1993positive}, if the quadratic module generated by the polynomials describing the semialgebraic set is archimedean, the hierarchy of semidefinite bounds converges from below to the minimum of the polynomial over this  semialgebraic set. 

In the free noncommutative context (i.e., without traces), a polynomial is positive semidefinite if and only if it can be written as a \emph{sum of hermitian squares} (SOHS)~\cite{Helton02,McCullSOS}.  
One can rely on such SOHS decompositions to perform  eigenvalue optimization of noncommutative polynomials over noncommutative semialgebraic sets, i.e., under noncommutative polynomial inequality constraints. 
The noncommutative analogues of Lasserre's hierarchy \cite{Helton04,navascues2008convergent,pironio2010convergent,cafuta2012constrained,nctrace}
 allow one to approximate as closely as desired the optimal value of such eigenvalue minimization problems.
In \cite{navascues2008convergent}, Navascu\'es, Pironio and Ac\'in provide a way to compute bounds on the maximal violation levels of Bell inequalities: they first reformulate the initial problem as an eigenvalue optimization one and then approximate its solution with a converging hierarchy of semidefinite programs, based on the noncommutative version of Putinar's Positivstellensatz due to Helton and McCullough \cite{Helton04}. 
This is the so-called Navascu\'es-Pironio-Ac\'in (NPA for short) hierarchy and can be viewed as the ``eigenvalue'' version of Lasserre's hierarchy.
This leads to a hierarchy of upper bounds on the maximum violation level of Bell inequalities (see also~\cite{doherty2008quantum,pal2009quantum}). 
Further extensions \cite{pironio2010convergent,cafuta2012constrained,nctrace} have been provided to optimize the trace of a given polynomial under positivity constraints.
$\ncsostools$  \cite{cafuta2011ncsostools,burgdorf16}  can compute lower bounds on minimal eigenvalues or traces of noncommutative polynomial objective functions over noncommutative semialgebraic sets. 

This work greatly extends these frameworks to the case of optimization problems involving {\em trace polynomials}, i.e., linear combinations of products of matrices and matrix traces.
A very simple example of such polynomial is $\Trace (A_1) \cdot \Trace (A_2^2) + (\Trace (A_1 A_2))^2$, where $A_1$ and $A_2$ are noncommutative variables, e.g., $A_1$ and $A_2$ can be both quantum physics operators. 
One important underlying motivation is that trace polynomials are involved in several problems arising from quantum information theory. 
For instance, \cite{fukuda2014asymptotically} presents a framework to obtain the limit output states for a large class of input states having  specific sets of parameters.
To obtain these limits, one needs to compute bounds for generalized traces of tensors.
One way to model such generalized traces is to consider a reformulation as an optimization problem involving trace polynomials.
In this problem, trace polynomials arise as cost functions but they can also appear in the constraints. 
Convex relaxations of trace polynomial problems  can be obtained as in the NPA hierarchy: 
one can associate a new variable to each word trace (e.g., $\Trace (A_1)$, $\Trace (A_1 A_2)$ and $\Trace (A_2^2)$ in the above example), then incorporate the initial constraints into the semidefinite matrix defined in the NPA hierarchy. 
Moreover the noncommuting operators, denoted by  $A_i, B_j, C_k$ in \cite{pozas2019bounding}, fulfill causal constraints, which leads to equality constraints such as
$$\Trace(A_{i_1} A_{i_2} \cdots A_{i_m} C_{k_1} C_{k_2} \cdots C_{k_m}) - \Trace(A_{i_1} A_{i_2} \cdots A_{i_m}) \Trace(C_{k_1} C_{k_2} \cdots C_{k_m}) = 0.$$
This results in a so-called {\em scalar extension} of the NPA hierarchy, which allows the authors to successfully identify correlations not attainable in the entanglement-swapping scenario.
However, \cite{pozas2019bounding} does not provide a proof of convergence for this hierarchy.
In \cite{huber2020positive}, the author focuses on the multilinear case and obtains a characterization of all multilinear equivariant trace polynomials which are positive on the positive cone. 
\revise{In contrast with \cite{huber2020positive}, we consider optimization of nonlinear trace polynomials over trace polynomial inequalities after assuming that the quadratic module
generated by the polynomials involved in the set of constraints is archimedean.}
In a closely related work in real algebraic geometry \cite{klep2018positive}, the first and third author derive several Positivstellens\"atze for trace polynomials positive on semialgebraic sets of \emph{fixed size} matrices. 
In particular, \cite{klep2018positive} establishes a Putinar-type Positivstellensatz stating that any positive polynomial admits a weighted SOHS decomposition without denominators.
In the dimension-free setting, finite von Neumann algebras and their tracial states provide a natural framework for studying tracial polynomial inequalities. This paper characterizes trace polynomials which are positive on tracial semialgebraic sets, where the initial polynomials and constraints involve freely noncommutative variables and traces, and the evaluations are performed on von Neumann algebras.

\subsection*{Contributions}

A {\em trace polynomial} is a polynomial in symmetric noncommutative variables $x_1,\dots,x_n$ and traces of their products.
Thus naturally each trace polynomial has an adjoint.
A {\em pure trace polynomial} is a trace polynomial that is made only of traces, i.e., has no free variables $x_j$. For instance,
the trace of a trace polynomial 
is a pure trace polynomial, e.g.
\[
\begin{split}
f&=x_1x_2x_1^2-\Trace(x_2)\Trace(x_1x_2)\Trace(x_1^2x_2)x_2x_1, \\
\Trace(f)&=\Trace(x_1^3x_2)-\Trace(x_2)\Trace(x_1x_2)^2\Trace(x_1^2x_2),\\
f^\star&= x_1^2x_2x_1-\Trace(x_2)\Trace(x_1x_2)  \Trace(x_1^2x_2) x_1x_2.
\end{split}
\]

Given a set of symmetric trace polynomials $S$, let $\cD_S$ be the set of all tuples $(X_1,\dots,X_n)$, where $X_j$ are operators from a finite von Neumann algebra with a given tracial state, that satisfy $s(X_1,\dots,X_n)\succeq0$ for all $s\in S$.

In Section~\ref{sec:noncyclic} we state and prove a pure trace variant of the 
Helton-McCullough \cite{Helton04}
noncommutative version of Putinar's Positivstellensatz \cite{Putinar1993positive}: we obtain a representation of pure trace polynomials positive on a set described by pure trace polynomial inequalities,
using weighted sums of squares. 
This first (noncylic) Positivstellensatz is valid under the classical assumption that the quadratic module generated by the polynomials involved in the set of constraints is archimedean (Corollary \ref{cor:psatz}).
Our proof relies on the classical Kadison–Dubois representation theorem (see e.g.~\cite{marshallbook}) and the Gelfand-Naimark-Segal (GNS) construction.

Then, we derive in Section \ref{sec:cyclic} a novel, cyclic Positivstellensatz for the more general case of trace polynomials which are positive on tracial semialgebraic sets.
A subset of symmetric trace polynomials $\cQ$ is a {\em cyclic quadratic module} if $1\in\cQ$, $\cQ+\cQ\subseteq \cQ$, $\Trace(\cQ)\subseteq\cQ$ and $h\cQ h^\star\subseteq\cQ$ for every trace polynomial $h$. In analogy with the commutative setting we say that $\cQ$ is {\em archimedean} if $N-x_1^2-\dots-x_n^2\in\cQ$ for some $N>0$. 
\revise{Observe that each $\cQ$ contains all sums of elements of the form
$\Trace(h_1h_1^\star)\cdots \Trace(h_\ell h_\ell^\star)h_0h_0^\star$
	for $h_i\in\T$. Lemma \ref{lemma:cycgen} below describes the smallest cyclic quadratic module $\cQ(S)$ containing a given set of generators $S\subseteq\T$.} 

\begin{thmA}[Corollary \ref{cor:linslackpsatz}]
Let $\cQ$ be an archimedean cyclic quadratic module, and let $a$ be a symmetric trace polynomial. The following are equivalent: 
\begin{enumerate}[\rm (i)]
	\item $a\succeq0$ on $\cD_{\cQ}$;
	\item for every $\varepsilon>0$ there exist univariate sums of squares $s_1,s_2\revise{\in\R[t]}$, 
	\revise{depending on $\varepsilon$}, such that
$$a=s_1(a)-s_2(a),\qquad \varepsilon-\Trace(s_2(a))\in \cQ.$$
\end{enumerate}
\end{thmA}

In Section \ref{sec:hierarchy}, we rely on \revise{this} Positivstellensatz to design a converging hierarchy of semidefinite relaxations to approximate from below the minimum of a pure trace polynomial under pure trace polynomial inequality constraints. 
An extension of this hierarchy to the more general case of trace polynomial constraints is presented in Section \ref{sec:constrained}.
\begin{thmA}[Corollary \ref{cor:pure_cvg}]
Let $S$ be a set of symmetric trace polynomials, and $a$ a pure trace polynomial. The Positivstellensatz-induced hierarchy of semidefinite programs produces a convergent increasing sequence with limit $\inf_{\cD_S}a$.
\end{thmA}

Along the way, we present in Section~\ref{sec:gns} a tracial variant of the finite-dimensional GNS construction under flatness and extremal assumptions.
We use it to obtain finite convergence of our hierarchy as well as exactness of the relaxed solution, and design an algorithm to extract minimizers.
\revise{
Finally, in Section~\ref{sec:new} 
we give a simple example demonstrating our theoretical results,
and present an application of our techniques to quantum information theory: we use 
tracial optimization to find upper bounds on violations of polynomial Bell inequalities.}

\section{Notation and basic definitions}
\label{sec:prelim}

We begin by introducing basic notions about noncommutative polynomials, trace polynomials, and semialgebraic sets that will be used throughout the paper.

\subsection{Noncommutative polynomials and trace polynomials}
\label{sec:prelim_nc}
Let us denote by $\Mbb_k$ (resp.~$\Sbb_k$) the space of all real (resp.~symmetric) matrices of order $k$.
The normalized trace of a matrix $A \in \Mbb_k$ is given by $\Trace A = \frac{1}{k} \sum_{i=1}^k A_{i,i}$.
For a fixed $n \in \N$, we consider a finite alphabet $x_1,\dots,x_n$ and generate all possible words of finite length in these letters. 
The empty word is denoted by 1. 
The resulting set of words is the {\em free monoid} $\mx$, with $\underline{x} = (x_1,\dots, x_n)$. 
We denote by $\RX$ the set of real polynomials in noncommutative variables, abbreviated as {\em nc polynomials}.
The algebra $\RX$ is equipped with the involution $\star$ that fixes $\R \cup \{x_1,\dots,x_n\}$ point-wise and reverses words, so that $\RX$ is the $\star$-algebra freely generated by $n$ symmetric letters $x_1,\dots,x_n$. 

We now introduce some algebraic terminology 
to deal with the trace, 
following \cite{P76} (see also \cite{KS17,klep2018positive}).
\revise{
Two words $u,v\in\mx$ are called $\star$-cyclically equivalent 
($u  \overset{\text{cyc}^\star}{\sim} v$)
if $v$ or $v^\star$ can be obtained from $u$ by cyclically rotating the letters in $u$. For example, all words of length $3$ are $\star$-cyclically equivalent, 
but $x_1x_2x_3x_4 \overset{\text{cyc}^\star}{\nsim} x_2x_1x_3x_4$.
}
We denote by $\skinnyT$ the commutative polynomial algebra in infinitely many variables
$\Trace (w)$ with $w \in \mx$, up to $\star$-cyclic equivalence, 
that is, $\skinnyT := \R[\Trace (w) \,,  w \in \mx /  \text{cyc}^\star]$.
We also let $\T := \skinnyT \mx$ be the
free $\skinnyT$-algebra on $\ux$.
Elements of $\skinnyT$ are called \emph{pure trace polynomials}, and elements
of $\fatT$ are \emph{trace polynomials}. For example, $t=\Trace(x_1^2)-\Trace(x_1)^2 \in \skinnyT$ and $x_1^2-\Trace(x_1) x_ 1 - 2 t \in \T = \skinnyT\langle x_1 \rangle$.
The involution on $\T$, denoted also by $\star$, fixes $\{ x_1,\dots, x_n\}\cup \skinnyT$ point-wise, and reverses words from $\mx$.
The set of all {\em symmetric elements} of $\T$ is defined as $\SymT := \{f \in \T : f = f^\star  \}$.
A linear functional $L : \T \to \R$ is said to be {\em tracial} if $L (\Trace(f)) = L(f)$ for all $f \in \T$.
We also consider the universal trace map $\tau$ defined by
\begin{align*}
\tau : & \ \T \to \skinnyT \,, \\
& \ f \mapsto \Trace(f) \,.
\end{align*}
A linear functional $L : \T \mapsto \R$ is tracial if and only if $L \circ \tau = L$.
Such an $L$ is determined by $L|_{\skinnyT} : \skinnyT \to \R$ being an (arbitrary) linear functional.
The functional $L$ is called {\em unital} if $L(1) = 1$ and is called {\em symmetric} if $L(f^\star) = L(f)$, for all $f$ belonging to the domain of $L$.
\subsection{Tracial semialgebraic sets and von Neumann algebras}
\label{sec:tracialvNa}

Given $S \subseteq \SymT$, the {\em matricial tracial semialgebraic set} $\mathcal{D}_S$ associated to $S$ is defined as follows:
\begin{align}
\label{eq:DSK}
\mathcal{D}_S := \bigcup_{k \in \N} \{ \underline{A} = (A_1,\dots,A_n) \in \Sbb_k^n : s(\underline{A}) \succeq 0 \ \text{for all}\ s\in S  \} \,.
\end{align}
While \eqref{eq:DSK} looks like a natural candidate for testing positivity of tracial polynomials, the failure of Connes' embedding conjecture \cite{CECfalse} hinders the existence of a reasonable Positivstellensatz for \eqref{eq:DSK} by \cite{CECsohs}. Instead of just matrices of all finite sizes, one is thus led to include bounded operators, similarly as in the trace-free setting \cite{Helton04}. Since we deal with tracial constraints, the considered bounded operators need to admit traces. The natural framework is therefore given by tracial von Neumann algebras, which we discuss next.

A real von Neumann algebra $\cF$ \cite{ARU97} is a unital, weakly closed, real, self-adjoint subalgebra of the (real) algebra of bounded linear operators on a complex Hilbert space, with the property $\cF\cap i\cF=\{0\}$. We restrict ourselves to separable Hilbert spaces, implying that all von Neumann algebras have separable preduals. Much of the structure theory of real von Neumann algebras can be transfered from complex von Neumann algebras \cite[Chapter~5]{Tak02}. Namely, the complexification of a real \vna yields a complex \vna with an involutory $*$-antiautomorphism; conversely, the fixed set of an involutory $*$-antiautomorphism on a complex \vna is a real \vna.
A real \vna is \emph{finite} if in its complexification, every isometry is a unitary. By \cite[Theorem 2.4]{Tak02}, a \vna is finite if and only if it admits sufficiently many normal tracial states, which will play an important role in this article. 

A (real) \vna is a factor if its center consists of only the (real) scalar operators. By \cite[Theorem 2.6]{Tak02}, a factor is finite if and only if it admits a faithful normal tracial state; in this case, such a state is unique, and is called the \emph{trace} of the factor. Finally, a \emph{$\II_1$-factor} is an infinite-dimensional finite factor (other finite factors are of type $\I_n$, 
which are $n\times n$ complex matrices in the complex setting, and 
$n\times n$ real matrices or $\tfrac{n}{2}\times\tfrac{n}{2}$ quaternion matrices
 in the real setting).
In this article we consider positivity on operator semialgebraic sets. These are defined as follows (cf. \cite[Definition~1.59]{burgdorf16}):
\begin{definition}
\label{def:DSII}
A tracial pair $(\cF,\tau)$ consists of a real finite von Neumann algebra $\cF$ and a faithful normal tracial state $\tau$ on $\cF$ \cite[Chapter~5]{Tak02}.

Given $S \subseteq \SymT$ let $\mathcal{D}_S^{\cF,\tau}$ be the set of all self-adjoint tuples $\uX = (X_1,\dots,X_n) \in \cF^n$ making $s(\uX)$ a positive semidefinite operator for every $s \in S$; here $\Trace$ is evaluated as $\tau$. 
The von Neumann semialgebraic set $\mathcal{D}_S^{\trvN}$ generated by $S$ is defined as
\[
\mathcal{D}_S^{\trvN} := \bigcup_{(\cF,\tau)} \mathcal{D}_S^{\cF,\tau} \,,
\]
where the union is over all tracial pairs $(\cF,\tau)$.
Analogously, we define
\[
\mathcal{D}_S^{\II_1} := \bigcup_{\cF} \mathcal{D}_S^{\cF} \,,
\]
where the union is over all $\II_1$-factors (which come equipped with unique traces).
\end{definition} 
Note that finiteness of $S$ is not needed at this stage. Unlike in the free case \cite{HKM11},
these tracial semialgebraic sets are closed neither under direct sums nor reducing subspace compressions; for example, if $s=\Trace(x_1)\Trace(x_2)$, then
\[
\begin{split}
s(3,1)>0 \quad\text{and}\quad s(-1,-2)>0, \quad\text{but}\quad s(3\oplus -1,1\oplus -2) <0; \\
s(-2\oplus 1,1\oplus -2) >0, \quad\text{but}\quad s(-2,1)<0 \quad\text{and}\quad s(1,-2)<0. 
\end{split}
\]
To sidestep this technical problem we make use
of the following well-known fact that is all but stated in \cite[Theorem 2.5]{Dyk94}.

\begin{proposition}\label{p:embed}
Every tracial pair embeds into a $\II_1$-factor.
\end{proposition}

\begin{proof}
Let $(\cF,\tau)$ be a tracial pair and $(\widetilde\cF,\widetilde\tau)$ its complexification. If $L(\Z)$ is the complex von Neumann group algebra of $\Z$ \cite[Definition 7.4]{Tak02}, the free product construction (see e.g. \cite[Section 1]{Dyk94}) along $\widetilde\tau$ and the trace on $L(\Z)$ yields the \vna $\widetilde\cF*L(\Z)$ with a normal faithful tracial state whose restriction to $\widetilde\cF$ equals $\widetilde\tau$. Also, both $L(\Z)$ and $\widetilde\cF$ admit natural involutory $\star$-antiautomorphisms, which induce an involutory $\star$-antiautomorphism $\widetilde\cF*L(\Z)$. Its fixed set is a real \vna algebra containing $(\cF,\tau)$, and it is a (real) $\II_1$-factor if $\widetilde\cF*L(\Z)$ is a (complex) $\II_1$-factor. If $\dim \widetilde\cF\le 2$, then $\widetilde\cF\in\{\C,\C\oplus\C \}$ and $\widetilde\cF*L(\Z)\in \{L(\Z),L(\Z_2*\Z) \}$ is a $\II_1$-factor. if $\dim \widetilde\cF\ge 3$, then $\widetilde\cF*L(\Z)$ is a  $\II_1$-factor by \cite[Theorem 2.5]{Dyk94} (and the proof of \cite[Lemma 2.9]{Dyk94}).
\end{proof}

\revise{
\begin{remark}
A reader might rightfully wonder why the setup is restricted to reals instead of complexes. Since every complex von Neumann algebra is a real von Neumann algebra and the real part of a complex tracial state is a real tracial state, all the conclusions in this paper also hold for evaluations in complex von Neumann algebras. Likewise one could consider complex trace polynomials, but the corresponding formalism for trace symbols would be more intricate (namely, trace symbols would not be fixed under the involution, so they would need to be split in real and imaginary part, and relations connecting both with respect to $\star$ would need to be imposed). However, with the view towards optimization and implementation using the standard semidefinite programming solvers it is more convenient to derive results within the real framework.
\end{remark}
}

\section{Non-cyclic Positivstellensatz for pure trace polynomials}
\label{sec:noncyclic}

In this section we provide our first Positivstellensatz, Theorem \ref{thm:nocyc}, for pure trace polynomials based on quadratic modules from real algebraic geometry \cite{marshallbook}.
\revise{
A subset $\cM \subseteq \skinnyT$ is called a quadratic module if $1 \in \cM$, $\cM + \cM \subseteq \cM$ and $a^2 \cM$ for all $a \in \skinnyT$.
}
Given an archimedean quadratic module $\cM \subseteq \skinnyT$ (in the usual commutative sense, meaning that for each $f\in\skinnyT$ there is $m>0$ such that $m\pm f\in\cM$), we consider the real points of the real spectrum $\operatorname{Sper}_{\cM}\skinnyT$, namely the set $\chi_{\cM}$ defined by\looseness=-1
\begin{align}
\label{eq:chiM}
\chi_{\cM} := \{\varphi : \skinnyT \to \R \mid \varphi \text{ homomorphism,}  \ \varphi(\cM) \subseteq \R_{\geq 0}, \ \varphi(1) = 1 \}.
\end{align}

The next proposition is the well-known Kadison-Dubois representation theorem, see e.g. \cite[Theorem 5.4.4]{marshallbook}.

\begin{proposition}
\label{prop:KD}
Let $\cM \subseteq \emph{\skinnyT}$ be an archimedean quadratic module.
Then, for all $a \in \emph{\skinnyT}$, one has
$$\forall \varphi \in \chi_{\cM}  \quad \varphi(a) \geq 0 \qquad \Leftrightarrow \qquad \forall \varepsilon > 0 \quad a + \varepsilon \in \cM.$$
\end{proposition}

A homomorphism $\varphi:\skinnyT\to\R$ is determined by the ``tracial moments'' $\varphi(\Trace(w))$ for $w\in\mx$. In this sense, the following variant of \cite[Theorem 1.3]{Had} is a solution of the tracial moment problem. In the given formulation, it is the dimension-free analog of the extension theorem \cite[Theorem 4.8]{klep2018positive}.

\begin{proposition}
\label{prop:moment}
Let $\varphi: \emph{\skinnyT}\to \R$ be a homomorphism. Then there are a tracial pair $(\cF,\tau)$ and $\uX=\uX^*\in\cF^n$ such that $\varphi(a)=a(\uX)$ for all $a\in \emph{\skinnyT}$ if and only if the following holds:
\begin{enumerate}[\rm (a)]
	\item $\varphi(\Trace (pp^\star))\ge0$ for all $p\in \RX$;
	\item $\liminf_{k\to\infty} \sqrt[2k]{\varphi(\Trace (x_j^{2k}))}<\infty$ for $j=1,\dots,n$. 
\end{enumerate}
\end{proposition}

\begin{proof}
$(\Rightarrow)$ This is trivial since $\tau(AA^*)\ge0$ and $|\tau(A^{2k})|\le \|A\|^{2k}$ for every $A\in\cF$ and $k\in\N$.

$(\Leftarrow)$ Denote
$$\alpha=\max\left\{1,\max_j\liminf_{k\to\infty} \sqrt[2k]{\varphi(\Trace (x_j^{2k}))}\right\}.$$
Let $\phi:\RX\to \R$ be a linear functional defined by
$$\phi(w):=\frac{\varphi(\Trace(w))}{\alpha^{|w|}}$$
for $w\in\mx$. Then $\phi$ is a symmetric tracial functional on $\RX$, $\phi(pp^\star)\ge0$ for every $p\in\RX$ and $\max_j\liminf_{k\to\infty} \phi(x_j^{2k})<\infty$. By \cite[Theorem 1.3]{Had} (or rather its real version) there is a tracial pair $(\cF,\tau)$ and a tuple of self-adjoint contractions $\uY \in\cF^n$ such that $\phi(p)=\tau(p(\uY))$ for all $p\in \RX$. Then $\uX=\alpha\uY$ satisfies $\varphi(\Trace (p))=\tau(p(\uX))$ for $p\in\RX$ and thus $\varphi(a)=a(\uX)$ for $a\in\skinnyT$.
\end{proof}

\begin{definition}
Given $S\subseteq \skinnyT$ and $N>0$ let 
\begin{align}
\label{eq:preSN}
S(N):=S
\cup\{\Trace (pp^\star)\mid p\in\RX \}
\cup \{N^{k}-\Trace (x_j^{2k})\mid  1\le j\le n,\ k\in\N \}
\subseteq\skinnyT.
\end{align}
For $S\subseteq\Sym\fatT$ let
\begin{align}
\label{eq:SN}
S[N]:=S\cup \{N-x_j^2\mid j=1,\dots,n\}\subseteq \fatT.
\end{align}
\end{definition}

\begin{lemma}\label{l:arch}
The quadratic module $\cM(S(N))\subseteq \emph{\skinnyT}$ is archimedean for every $S, N$.
\end{lemma}

\begin{proof}
We need to show that for every $w\in\mx$ there exists $m>0$ such that
\begin{equation}\label{e:bd}
m\pm\Trace(w)\in \cM(S(N)).
\end{equation}
Write $w=x_{i_1}^{k_1}\cdots x_{i_\ell}^{k_\ell}$ for $i_1\neq i_2\neq\cdots \neq i_\ell$; we prove \eqref{e:bd} by induction on $\ell$. If $\ell=1$, \eqref{e:bd} holds because
$$N^k+1\pm 2\Trace(x_j^k) = \big(N^k-\Trace(x_j^{2k})\big)+\Trace\big((x_j^k\pm 1)^2\big).$$
For $\ell>1$ denote $\lambda=\lfloor \frac{\ell}{2}\rfloor$,
 and let $w_1=x_{i_1}^{k_1}\cdots x_{i_\lambda}^{k_\lambda}$ and $w_2=x_{i_{\lambda+1}}^{k_{\lambda+1}}\cdots x_{i_\ell}^{k_\ell}$. Then
$$\Trace(w_1w_1^\star) = \Trace\left(
x_{i_1}^{2k_1}x_{i_2}^{k_2}\cdots
x_{i_{\lambda-1}}^{k_{\lambda-1}}x_{i_\lambda}^{2k_\lambda}x_{i_{\lambda-1}}^{k_{\lambda-1}}
\cdots x_{i_2}^{k_2}
\right)$$
and similarly for $\Trace(w_2w_2^\star)$; note that $2(\lambda-1),2(\ell-\lambda-1)<\ell$. Hence by the induction hypothesis there exist $m_1,m_2>0$ such that $m_i- \Trace(w_iw_i^\star)\in \cM(S(N))$. Then\looseness=-1
$$(m_1+m_2) \pm 2\Trace(w) =
\left(m_1-\Trace(w_1w_1^\star)\right) + \left(m_2-\Trace(w_2w_2^\star)\right)
+ \Trace\left((w_1\pm w_2^\star)(w_1^\star\pm w_2) \right)$$
lies in $\cM(S(N))$.
\end{proof}

\revise{
Recall that the Helton-McCullough archimedean Positivstellensatz \cite{Helton04} states that any  polynomial in noncommutative variables positive on a basic semialgebraic set belongs to the quadratic module generated by the polynomials describing this set, under the assumption that this module is archimedean. 
}
We are now ready to prove our first theorem, the purely tracial analog of this noncommutative Helton-McCullough Positivstellensatz.

\begin{theorem}\label{thm:nocyc}
Let $S\subseteq \emph{\skinnyT}$ and $N>0$ 
be given.
Then for $a\in\emph{\skinnyT}$ the following are equivalent:
\begin{enumerate}[\rm (i)]
\item\label{it:i2} $a(\uX)\geq 0$ for all 
$\uX \in 
\mathcal{D}_{S[N]}^{\trvN}$;\vspace{0.25ex}
\item\label{it:i3} $a(\uX)\geq 0$ for all $\uX \in \mathcal{D}_{S[N]}^{\II_1}$;\vspace{0.25ex}
\item\label{it:i1} $a+\varepsilon \in \cM(S(N))$ for all $\varepsilon>0$.
\end{enumerate}
\end{theorem}

\begin{proof}
\ref{it:i2}$\Leftrightarrow$\ref{it:i3} holds by Proposition \ref{p:embed}.
\ref{it:i1}$\Rightarrow$\ref{it:i2} If $\uX \in \mathcal{D}_{S[N]}^{\cF,\tau}$, then
$$s(\uX)\ge0,\qquad \tau (p(X)p(X)^*)\ge0,\qquad \tau(X_j^{2k} )\le N^{k}$$
for all $s\in S$, $p\in\RX$, $1\le j\le n$ and $k\in\N$, so $a(\uX)\ge0$.

\ref{it:i2}$\Rightarrow$\ref{it:i1} Suppose $a+\varepsilon \notin \cM(S(N))$ for some $\varepsilon>0$. By Proposition \ref{prop:KD}, there exists a unital homomorphism $\varphi: \skinnyT\to\R$ with $\varphi(\cM(S(N)))\subseteq \R_{\ge0}$ and $\varphi(a)<0$. Hence
$$\varphi(\Trace (pp^\star))\ge0,\qquad \varphi(\Trace (x_j^{2k}))\le N^k$$
for all $p\in\RX$, $1\le j\le n$ and $k\in\N$. Hence by Proposition \ref{prop:moment} there exist a tracial pair $(\cF,\tau)$ and $\uX\in\cF^n$ such that $\varphi(b)=b(\uX)$ for all $b\in \skinnyT$. Moreover, the proof of Proposition \ref{prop:moment} implies $\|X_j\|\le\sqrt N$ for $1\le j\le n$. Furthermore, $s(\uX)=\varphi(s)\ge0$ for every $s\in S$ implies $\uX \in \mathcal{D}_{S[N]}^{\cF,\tau}$. Finally, $a(\uX)=\varphi(a)<0$.
\end{proof}

Since $\cM(S(N_1))\subseteq \cM(S(N_2))$ for $N_1\ge N_2$, we obtain the following:

\begin{corollary}\label{cor:psatz}
Let $S\subseteq \emph{\skinnyT}$ and $a\in \emph{\skinnyT}$. The following are equivalent:
\begin{enumerate}[\rm (i)]
	\item $a(\uX)\geq 0$ for all 
$\uX \in 
\mathcal{D}_{S}^{\trvN}$;\vspace{0.25ex}
	\item $a(\uX)\geq 0$ for all $\uX \in \mathcal{D}_{S}^{\II_1}$;\vspace{0.25ex}
	\item $a+\varepsilon \in \cM(S(N))$ for all $\varepsilon>0$ and $N\in\N$.
\end{enumerate}
\end{corollary}

\section{Cyclic Positivstellensatz for trace polynomials}
\label{sec:cyclic}
In this section we prove a Positivstellensatz for trace polynomials that is 
less inspired by the commutative theory than 
the one from Section \ref{sec:noncyclic} and relies more on the tracial structure of trace polynomials.
First we introduce the notion of a cyclic quadratic module.
A subset $\cQ \subseteq \SymT$ is called a {\em cyclic quadratic module} if
$$1 \in \cQ,\ \cQ+\cQ\subseteq \cQ,\  
a^\star \cQ a \subseteq \cQ\ \forall a\in\fatT,\ 
\Trace(\cQ)\subseteq\cQ.$$
Given $S\subseteq\T$ let $\cQ(S)$ be the cyclic quadratic module generated by $S$, i.e., the smallest cyclic quadratic module in $\T$ containing $S$.
A cyclic quadratic module $\cQ$ is called {\em archimedean} if for all $a \in \SymT$ there exists $N>0$ such that $N - a \in \cQ$.
We start with a few preliminary results.

\begin{lemma}
\label{lemma:cycgen}
Let $S\subseteq\T$.
\begin{enumerate}[\rm (1)]
	\item Elements of $ \cQ(\emptyset)$ are precisely sums of
	$$\Trace(h_1h_1^\star)\cdots \Trace(h_\ell h_\ell^\star)h_0h_0^\star$$
	for $h_i\in\T$.
	\item Elements of $\cQ(S)$ are precisely sums of
	$$q_1,\quad h_1s_1h_1^\star,\quad \Trace(h_2s_2h_2^\star)q_2$$
	for $h_i\in \T$, $q_i\in \cQ(\emptyset)$, $s_i\in S$.
	\item Elements of $\Trace(\cQ(S))=\cQ(S)\cap\skinnyT$ are precisely sums of
	$$\Trace(h_1h_1^\star)\cdots \Trace(h_\ell h_\ell^\star)\Trace(h_0 s h_0^\star )$$
	for $h_i\in\T$ and $s\in S$.
\end{enumerate}
\end{lemma}

\begin{proof}
Straightforward.
\end{proof}

\revise{
Note that expressions such as $\Trace(h_1 s_1 h_1^\star) \Trace(h_2 s_2 h_2^\star)$ for $h_i\in \T$ and $s_i \in S$ do not belong to $\cQ(S)$. 
{We emphasize that computing such product representations in our context would be very hard in practice. 
Indeed, even if one bounds the degrees of the $h_j$, computing their coefficients boils down to solving a nonlinear semidefinite program, which is impractical.}
Second, even in the classical commutative  case 
quadratic modules (such as those appearing in Putinar's Positivstellensatz) are not closed under multiplication, by contrast with Schm\"ugden type representations \cite{Schmudgen91sos} in which one allows multiplication of polynomials involved in the set of constraints. 
 Admitting products of constraints in sums of squares positivity certificates increases computational cost only modestly, but
 from a theoretical
viewpoint yields little to no advantages (in the archimedean case),
which is why quadratic modules are preferred from a practical point of view.}
\begin{proposition}
\label{prop:cyclicqmarch}
A cyclic quadratic module $\cQ$ is archimedean if and only if there exists $N \in \N$ such that $N - \sum_{i=1}^n x_j^2 \in \cQ$.
\end{proposition}
\begin{proof}
($\Rightarrow$) is obvious. For the converse assume $N - \sum_{i=1}^n x_j^2 \in \cQ$ for some $N\in\N$. Then
the set $\cQ \cap \RX$ is an archimedean quadratic module. Thus, for all $a=a^\star\in \RX$ there exists  $N \in \N$ such that 
\begin{align}
\label{eq:cyclicqmarch1}
N - a \in \cQ \cap \RX \,.
\end{align}
In addition, the set $H$ of bounded elements, defined by
\[
H = \{a \in \T \mid \exists N \in \N \text{ s.t. } N -a^\star a \in \cQ \} \,,
\]
is closed under involution, addition, subtraction and multiplication, i.e., is a $\star$-subalgebra of $\T$ \cite{vidav}.
A symmetric element $b\in\T$ is in $H$ if and only if there is some $N\in\N$
with $N\pm b\in \cQ$.

For every $a\in\ax$ we have
\begin{equation}\label{eq:cs1}
\Trace(aa^\star)-\Trace(a)^2 = \Trace\large((a-\Trace(a))(a-\Trace(a))^\star\large)\in \cQ.
\end{equation}
By \eqref{eq:cyclicqmarch1} and the fact that $\cQ$ is cyclic, 
there is some $N\in\N$ with $N-\Trace(aa^\star)\in \cQ$.
Adding this to \eqref{eq:cs1} yields
$N-\Trace(a)^2\in \cQ$. Thus, by the definition of $H$, $\Trace(a)\in \cQ$.
The desired result now follows since $H$ is a subalgebra of $\T$.
\end{proof}
\begin{proposition}\label{p:poly}
Let $(\cF,\tau)$ be a tracial pair and $X=X^*\in\cF$. The following are equivalent:
\begin{enumerate}[\rm (i)]
	\item $X\succeq0$;
	\item $\tau(XY)\ge0$ for all positive semidefinite contractions $Y\in\cF$;
	\item $\tau(Xp(X)^2)\ge0$ for all $p\in\R[t]$.
\end{enumerate}
\end{proposition}

\begin{proof}
(i)$\Rightarrow$(ii) is clear and (ii)$\Rightarrow$(iii) holds since for $p(X)\neq0$,
$$\tau(Xp(X)^2)=\|p(X)\|^2\,\tau\Big(X\left(\|p(X)\|^{-2}p(X)^2\right)\Big)$$
and $\|p(X)\|^{-2}p(X)^2$ is a positive semidefinite contraction. To prove (iii)$\Rightarrow$(i), let $\cF_1\subseteq \cF$ be the weak operator topology closure of the algebra generated by $X$. Then $\cF_1$ is an abelian von Neumann algebra and therefore $(\cF_1,\tau|_{\cF_1}) \cong (L^\infty(\cX,\mu),\int\cdot\,d\mu)$ for some standard measure space $(\cX, \mu)$ by \cite[Theorem III.1.18]{Tak02}. For $f\in L^\infty(\cX,\mu)$ we have $f\succeq0$ if and only if $\int fg^2\,d\mu\ge0$ for all $g\in L^\infty(\cX,\mu)$. If $f$ is the image of $X$ under the above isomorphism, then $\{p(f)\mid p\in\R[t]\}$ is dense in $\cF_1$. Hence $\int fp(f)^2\,d\mu\ge0$ for all univariate polynomials $p$ implies $f\succeq0$, so (iii)$\Rightarrow$(i) holds.
\end{proof}

The following is the cyclic version of the Helton-McCullough theorem \cite{Helton04}. Note that while the constraints in Theorem \ref{thm:psatz} are arbitrary trace polynomials, the objective function needs to be a pure trace polynomial. A direct analog for non-pure trace objective polynomials fails, see Example \ref{ex:unclean} below.

\begin{theorem}
\label{thm:psatz}
Let $\cQ \subseteq \Sym \T$ be an archimedean cyclic quadratic module and $a \in \emph{\skinnyT}$. The following are equivalent:
\begin{enumerate}[\rm (i)]
\item\label{it:j2} $a(\uX)\geq 0$ for all $\uX \in \mathcal{D}_\cQ^{\trvN}$;\vspace{0.25ex}
\item\label{it:j3} $a(\uX)\geq 0$ for all
$\uX \in \mathcal{D}_\cQ^{\II_1}$;\vspace{0.25ex}
\item\label{it:j1} $a+\varepsilon \in \cQ$ for all $\varepsilon>0$.
	\end{enumerate}
\end{theorem}

\begin{proof}
Implications \ref{it:j1}$\Rightarrow$\ref{it:j2}$\Rightarrow$\ref{it:j3} are straightforward, so consider \ref{it:j3}$\Rightarrow$\ref{it:j1}. By Proposition \ref{p:poly} we have $\mathcal{D}_\cQ^{\II_1}=\mathcal{D}_{\Trace(\cQ)}^{\II_1}$. Since $\cQ$ is archimedean, there exists $N>0$ such that $N-x_j^2\in \cQ$ for $j=1,\dots,n$. Consequently
$$\mathcal{D}_\cQ^{\II_1}=\mathcal{D}_{\cQ[N]}^{\II_1}=\mathcal{D}_{\Trace(\cQ)[N]}^{\II_1}.$$
For every $j,k$ we have
$$N^k-\Trace(x_j^{2k})=\Trace\left(\left(\sum_{i=0}^{k-1}N^{k-1-i}x_j^{2i}\right)(N-x_j^2)\right) \in \cQ.$$
Consequently $\Trace(\cQ)(N)\subseteq \cQ$ and thus $\cM(\Trace(\cQ)(N)) \subseteq \cQ$. Then $a+\varepsilon \in \cQ$ for all $\varepsilon>0$ by Theorem \ref{thm:nocyc}.
\end{proof}

For the reader unfamiliar with real algebraic geometry and noncommutative moment problems, we present a self-contained proof of Theorem \ref{thm:psatz} relying only on convex separation results and basic properties of von Neumann algebras in Appendix \ref{app}.

\revise{
Having reached this point, it is tempting to think that membership in $\cQ$ is also enough to characterize all positive tracial polynomials, not just the pure ones. As it turns out, this intuition is wrong: in the following lines, we provide a simple counterexample.}
Given a set of symmetric polynomials $S\subseteq\RX$ let $\cM(S)$ denote the (free) quadratic module generated by $S$ \cite[Section 1.4]{burgdorf16}. Hence $\cM(S)$ is the smallest set that contains $S\cup\{1\}$, is closed under addition, and $f\in\cM(S)$ implies $hfh^\star\in\cM(S)$ for every $h\in\RX$.

\begin{lemma}\label{l:split}
Let $S_1\subseteq\skinnyT$ and $S_2\subseteq \RX$. If $s(0)\ge0$ for all $s\in S_1$, then
$$\cQ(S_1\cup S_2)\cap\RX = \cM(S_2).$$
\end{lemma}

\begin{proof}
Let $\pi_1:\skinnyT\to\R$ be given by $\pi_1(a)=a(0)$, and consider the homomorphism $\pi=\pi_1\otimes\id_{\RX}:\T\to\RX$. Then $\pi(\cQ(S_1\cup S_2))= \cM(S_2)$ because $\pi_1(S_1)\subseteq\R_{\ge0}$ and $\pi|_{\RX}=\id_{\RX}$. So the statement follows.
\end{proof}

\begin{example}\label{ex:unclean}
Let $n=1$. Let $\cQ$ be the archimedean cyclic quadratic module in $\Sym\T$ generated by
$$\{1-x_1^2\}\cup \{\Trace(x_1p^2(x_1))\mid p\in\R[t] \}.$$
By Proposition \ref{p:poly}, $X_1\in \mathcal{D}_\cQ^{\cF,\tau}$ implies $X_1\succeq0$ for any tracial pair $(\cF,\tau)$. On the other hand, if $\varepsilon\in[0,1)$ then $x_1+\varepsilon \notin \cM (\{1-x_1^2\})$ and therefore $x_1+\varepsilon \notin \cQ$ by Lemma \ref{l:split}.
\end{example}
\revise{
We emphasize that Example  \ref{ex:unclean} shows that Theorem \ref{thm:psatz} does not
hold in general for non-pure objectives.}
To mitigate the absence of a non-pure analog of Theorem \ref{thm:psatz}, we require the following technical lemma.

\begin{lemma}\label{l:univar}
Let $\varepsilon>0$ and $n=\lceil 1/\varepsilon \rceil$. If $s_2=\frac{\varepsilon}{2}(t-1)^{2n}$ and $s_1=s_2+t$, then
\begin{enumerate}[\rm (a)]
	\item $s_1$ is positive on $\R$, and thus a sum of (two) squares in $\R[t]$;
	\item $\frac{\varepsilon}{2}-s_2$ is nonnegative on $[0,1]$, and thus an element of $\cM(\{t,1-t\})$.
\end{enumerate}
\end{lemma}

\begin{proof}
(a) Clearly $s_1(\alpha)>0$ for $\alpha\ge0$. Since $\frac{{\rm d}s_2}{{\rm d} \revise{\alpha}}(\alpha)=\varepsilon n (\alpha-1)^{2n-1}<-1$ for every $\alpha\le 0$, we also have $s_2(\alpha)>-\alpha$ for $\alpha\le0$. So $s_1$ is positive on $\R$ and thus a sum of two squares by an easy application of the fundamental theorem of algebra (see e.g. \cite[Proposition 1.2.1]{marshallbook}).

(b) $\frac{\varepsilon}{2}-s_2$ is nonnegative on $[0,1]$ because $(\alpha-1)^{2n}\le 1$ for $\alpha\in[0,1]$. Since $t(1-t)=(1-t)^2t+t^2(1-t) \in \cM(\{t,1-t\})$, a result of Fekete \cite[Problem VI.46]{PS76} (see \cite[Proposition 2.7.3]{marshallbook} for a modern treatment) implies $\frac{\varepsilon}{2}-s_2\in \cM(\{t,1-t\})$. 
\end{proof}

Although the tracial version of the Helton-McCullough Positivstellensatz \cite{Helton04} fails, we have the following positivity certificate for non-pure trace polynomials.

\begin{corollary}\label{cor:linslackpsatz}
Let $\cQ \subseteq \Sym \T$ be an archimedean cyclic quadratic module and $a \in \Sym\T$. The following are equivalent:
\begin{enumerate}[\rm (i)]
\item\label{it:two1} $a(\uX)\succeq 0$ for all $\uX \in \mathcal{D}_\cQ^{\trvN}$;\vspace{0.25ex}
\item\label{it:two2} $a(\uX)\succeq 0$ for all 
$\uX \in \mathcal{D}_\cQ^{\II_1}$;\vspace{0.25ex}
\item\label{it:three1} for every $\varepsilon>0$, there exist sums of (two) squares $s_1,s_2\in\R[t]$ such that
\begin{equation}\label{e:woy}
a=s_1(a)-s_2(a),\qquad \varepsilon-\Trace(s_2(a))\in \cQ;
\end{equation}
\item\label{it:one1} for every $\varepsilon>0$, there exist sums of (two) squares $s_1,s_2\in\R[t]$ and $q\in \cQ$ such that
\begin{equation}\label{e:wy}
\Trace(ay)+\varepsilon =\Trace(s_1(a)y+s_2(a)(1-y))+q
\end{equation}
where $y$ is an auxiliary symmetric free variable. That is, $\Trace(ay)+\varepsilon$ is in the cyclic quadratic module generated by $\cQ,y,1-y$ (inside the free trace ring generated by $\ux,y$).
\end{enumerate}
\end{corollary}

\begin{proof}
\ref{it:two2}$\Leftrightarrow$\ref{it:two1} holds by Proposition \ref{p:embed}.
\revise{To prove \ref{it:three1}$\Rightarrow$\ref{it:one1}, note that \eqref{e:woy} implies (after multiplication by $y$ and taking the trace) that $\Trace(ay)+\varepsilon =\Trace(s_1(a)y + s_2(a)(1-y)) + \varepsilon -  \Trace(s_2(a))$. The implication follows by taking $q=\varepsilon-\Trace(s_2(a))\in \cQ$.}
\ref{it:one1}$\Rightarrow$\ref{it:two1} Let $(\cF,\tau)$ be a tracial pair. If $\Trace(ay)+\varepsilon$ belongs to the cyclic quadratic module generated by $\cQ,y,1-y$ for every $\varepsilon>0$, then $\Trace(a(\uX)Y)\ge0$ for all $\uX \in \mathcal{D}_\cQ^{\cF,\tau}$ and positive semidefinite contractions $Y\in\cF$. Therefore $a(\uX)\succeq0$ for all $\uX \in \mathcal{D}_\cQ^{\cF,\tau}$ by Proposition \ref{p:poly}.

\ref{it:two1}$\Rightarrow$\ref{it:three1} Since $\cQ$ is archimedean, there exists $N>0$ such that $N-a\in \cQ$. After rescaling $a$ we can without loss of generality assume that $1-a\in \cQ$. Suppose that \ref{it:two1} holds. Given an arbitrary $\varepsilon>0$ let $s_1,s_2\in\R[t]$ be sums of squares as in Lemma \ref{l:univar}. Then there are sums of squares $s_3,s_4,s_5\in\R[t]$ such that
$$s_1-s_2=t,\qquad \tfrac{\varepsilon}{2}-s_2 =s_3+s_4t+s_5(1-t).$$
By Proposition \ref{p:poly} and Theorem \ref{thm:psatz} we have $\Trace(s_4(a)a)+\tfrac{\varepsilon}{2}\in \cQ$.
Hence
$$\big(s_3(a)+\Trace(s_4(a)a)+\tfrac{\varepsilon}{2}\big)+s_5(a)(1-a)\in \cQ$$
and therefore
\[a=s_1(a)-s_2(a),\qquad \varepsilon-\Trace(s_2(a))\in \cQ. \qedhere\]
\end{proof}

\section{SDP hierarchy for trace optimization}
\label{sec:hierarchy}

In this section we apply Theorem \ref{thm:nocyc} to optimization of pure trace objective functions subject to (pure) trace constraints and a norm  boundedness condition. Doing so, we obtain a converging hierarchy of SDP relaxations in Section \ref{sec:pure_constrained}. 
When flatness occurs in this hierarchy, one can extract a finite-dimensional minimizer as shown in Section \ref{sec:gns}. 
Finally, we apply Proposition \ref{p:poly} to handle the more general case of trace polynomials subject to trace constraints and a norm boundedness condition in Section \ref{sec:constrained}.

We define the set of \emph{tracial words} (abbreviated as $\fatT$-words) by
$\{\prod_i\Trace(u_i) v \mid u_i,v \in \mx \}$,
which is a subset of $\T$.
The set of \emph{pure trace words} (abbreviated as $\skinnyT$-words) is the subset of $\fatT$-words belonging to $\skinnyT$.
For instance,  $\Trace(x_1)^2$ is a  $\skinnyT$-word and $\Trace(x_1) x_1$ is a $\fatT$-word.
For $u_i,v \in \mx$, we define the {\em tracial degree} of $\prod_i\Trace(u_i) v$ as the sum of the degrees of the $u_i$ and the degree of $v$.
The tracial degree of a trace polynomial $f \in \fatT$ is the length of the longest tracial word involved in $f$ up to cyclic equivalence.
Let us denote by $\W^{\fatT}_d$ (resp. $\W^{\skinnyT}_d$) the vector of all $\fatT$-words (resp. \revise{$\skinnyT$}-words) \revise{of tracial degree at most $d$} w.r.t.~to the \revise{degree} lexicographic order.
Finally, let $\TX_d$ (resp. $\skinnyT_d$) denote the span of entries of $\W^{\fatT}_d$ (resp. $\W^{\skinnyT}_d$) in $\TX$ (resp. $\skinnyT$), and let $\bsigma^{\fatT}(n,d)$ (resp. $\bsigma^{\skinnyT}(n,d)$) the dimension of $\TX_d$ (resp. $\skinnyT_d$), that is, the length of $\W^{\fatT}_d$ (resp. $\W^{\skinnyT}_d$).

\revise{Precise values of $\bsigma^{\skinnyT}(n,d)\le\bsigma^{\fatT}(n,d)$ are related to bracelet counting in combinatorics. To get a crude estimate, notice that the number of tracial words of degree $d$ is at least $n^d$ and at most $n^d\cdot 2^d$. Thus
$$\frac{n^{d+1}-1}{n-1}\le\bsigma^{\fatT}(n,d)\le \frac{(2n)^{d+1}-1}{2n-1}.$$
}

We introduce the notion of trace Hankel and (pure) trace localizing matrices, which can be viewed as tracial analogs of the noncommutative localizing and Hankel matrices (see e.g. \cite[Lemma~1.44]{burgdorf16}).
Given $s \in \fatT$, let us denote $d_s := \lceil \deg s / 2 \rceil$. To $s$ and a linear functional $L : \skinnyT_{2 d} \to \R$, one associates the following three matrices:
\begin{enumerate}[(a)]
\item the {\em tracial Hankel matrix} $\M_d^{\fatT}(L)$ is the symmetric matrix of size $\bsigma^{\fatT} (n,d)$,  indexed by $\fatT$-words $u,v \in \TX_{d}$, with $(\M_d^{\fatT}(L) )_{u,v} = L (\Trace(u^\star v))$;
\item if $s \in \skinnyT$, then the {\em pure trace localizing matrix} $\M^{\skinnyT}_{d -  d_s}(s \, L)$ is the symmetric matrix of size $\bsigma^{\skinnyT} (n,d - d_s)$, indexed by $\skinnyT$-words  $u,v \in \skinnyT_{d - d_s}$, with $(\M_{d - d_s}^{\skinnyT}(s \, L))_{u,v} = L (u v s)$;
\item the {\em trace localizing matrix} $\M^{\fatT}_{d -  d_s}(s \, L)$ is the symmetric matrix of size $\bsigma^{\fatT} (n,d - d_s)$, indexed by $\fatT$-words  $u,v \in \fatT_{d - d_s}$, with $(\M_{d - d_s}^{\fatT}(s \, L))_{u,v} = L ( \Trace(u^\star s v))$.
\end{enumerate}
\begin{definition}
\label{def:hankel_condition}
A matrix $\M$ indexed by $\fatT$-words of degree  $\leq d$ satisfies the tracial Hankel condition if and only if 
\begin{align}
\label{eq:hankel_condition}
\M_{u,v} = \M_{w,z} \text{ whenever } \Trace(u^\star v) = \Trace (w^\star z) \,.
\end{align}
\end{definition}
\begin{remark}
\label{rk:Hankelbij}
Linear functionals on $\skinnyT_{2 d}$ and matrices from $\Sbb_{\bsigma^{\fatT}(n,d)}$ satisfying the tracial Hankel condition \eqref{eq:hankel_condition} are in bijective correspondence.
To a linear functional $L : \skinnyT_{2 d} \to \R$, one can assign the matrix $\M_d^{\fatT}(L)$, defined by $(\M_d^{\fatT}(L) )_{u,v} = L (\revise{\Trace}(u^\star v))$, satisfying the tracial Hankel condition, and vice versa. 
\end{remark}
One can relate the positivity of $L$ and the positive semidefiniteness of its tracial Hankel matrix $\M_d^{\fatT}(L)$. The proof of the following lemma is straightforward and analogous to its free counterpart \cite[Lemma 1.44]{burgdorf16}.
\begin{lemma}
\label{lemma:pureHankel}
Given a linear functional $L : \emph{T}_{2 d} \to \R$, one has $L(\Trace(f^\star f)) \geq 0$ for all $f \in \fatT_{d}$, if and only if, $\M_d^\T (L) \succeq 0$. 
Given $s \in \emph{T}$, one has $L(a^2 s) \geq 0$ for all $a \in \emph{T}_{d - d_s }$, if and only if, $\M_{d - d_s}^{\emph{\skinnyT}} (s \, L) \succeq 0$.
Given $s \in \fatT$, one has $L(\Trace(f^\star \, s \, f)) \geq 0$ for all $f \in \fatT_{d - d_s }$, if and only if, $\M_{d - d_s}^{\fatT} (s \, L) \succeq 0$.
\end{lemma}
\subsection{SDP hierarchy for pure trace polynomial optimization}
\label{sec:pure_constrained}

For a finite $S\subseteq\skinnyT$, $N>0$ and $d\in\N$ define
\begin{equation}
\cM(S(N))_d:=\left\{
\sum_{i=1}^{K} a_i^2s_i\mid K\in\N,\ a_i\in\skinnyT,\ s_i \in S(N),\ \deg(a_i^2s_i)\le 2d
\right\}.
\end{equation}
Given $b \in \skinnyT$ and $p \in \RX$, note that $b^2 \Trace(p p^\star) = \Trace ((b p) (b p)^\star)$.
Therefore, elements of $\cM(S(N))_d$ correspond to sums of elements of the form
\begin{align}
\label{eq:elementsM}
a_1^2 s \,, \quad a_2^2 \big(N^{k} - \Trace (x_j^{2k})\big) \,, \quad \Trace (f f^\star) \,,
\end{align}
which are of degree at most $2d$, for $a_i \in \skinnyT$, $s \in S$, $1\le j\le n$, $k\in\N$, $f \in \fatT$.

Given a pure trace polynomial $a\in\skinnyT$,  one can then use $\cM(S(N))_d$ for $d=1,2,\dots$ to design a hierarchy of semidefinite relaxations for minimizing $a\in\skinnyT$ over the von Neumann semialgebraic sets $\mathcal{D}_{S[N]}^{\trvN}$ or $\mathcal{D}_{S[N]}^{\twoone}$.

Let us define  $a_{\min}$ and $a_{\min}^{\twoone}$ as follows:
\begin{align}
a_{\min} &:= \inf \{ a(\uA) \mid \uA \in \mathcal{D}_{S[N]} \} \,, \\ \label{eq:pure_constr}
a_{\min}^{\twoone} &:= \inf \{ a(\uA) \mid \uA \in \mathcal{D}_{S[N]}^{\twoone} \} 
=\inf \{ a(\uA) \mid \uA \in \mathcal{D}_{S[N]}^{\trvN} \}\,.
\end{align}
Here the equality in \eqref{eq:pure_constr} holds by Proposition \ref{p:embed}.
Since $\mathcal{D}_{S[N]}$ is a subset of $\mathcal{D}_{S[N]}^{\trvN}$, one has $a_{\min}^{\twoone} \leq a_{\min}$.
Let $d_{\min} := \max \{ d_s : s \in \{a \} \cup  S(N) \}$.
Then, one can under-approximate $a_{\min}^{\twoone}$ via the following hierarchy of SDP programs, indexed by $d \geq d_{\min}$: 
\begin{align}
\label{eq:pure_constr_dual}
a_{\min,d} = \sup \{ m \mid a - m \in \cM(S(N))_d \} \,.
\end{align}
\begin{lemma}
\label{lemma:pure_dual}
The dual of \eqref{eq:pure_constr_dual} is the following SDP problem:
\begin{equation}
\label{eq:pure_constr_primal}
\begin{aligned}
\inf_{\substack{L : \emph{T}_{2 d} \to \R \\ L \emph{ linear}}} \quad  & L(a)  \\	
\emph{s.t.} 
\quad & (\M_d^\fatT(L))_{u,v} = (\M_d^\fatT(L))_{w,z}  \,, \quad \text{whenever } \Trace(u^\star v) =  \Trace( w^\star z) \,, \\
\quad & (\M_d^\fatT(L))_{1,1} = 1 \,, \\
\quad & \M_d^{\fatT}(L) \succeq 0 \,, \\
\quad & \M_{d - d_s}^{\emph{T}}(s \, L) \succeq 0  \,,  \quad \text{for all }  s \in S  \,,\\
\quad & \M_{d - k}^{\emph{T}}( (N^{k} - \Trace (x_j^{2k}) ) \, L) \succeq 0  \,,  \quad \text{for all }  j=1,\dots,n\,, k\le d  \,.
\end{aligned}
\end{equation}
\end{lemma}
\begin{proof}
Let us denote by $(\cM(S(N))_d)^{\vee}$ the dual space of $\cM(S(N))_d$.
From \eqref{eq:elementsM}, one has
\begin{align*}
(\cM(S(N))_d)^{\vee} =  \{ &  L : \text{T}_{2d} \to \R \mid L  \text{ linear}, \\ 
& L(a^2 s) \geq 0 \,, \forall s \in S \,, \forall a \in \skinnyT_{d - d_s} \,, \\
& L(a^2  (N^{k} - \Trace (x_j^{2k} )\revise{)} \geq 0 \,,  \forall j=1,\dots,n \,, \forall k\le d \,, \forall a \in \skinnyT_{d - k} \,, \\
& L (\Trace(f f^\star)) \geq 0 \,, \forall f \in \TX_d \}
\end{align*}
By using a standard Lagrange duality approach, we obtain the dual of SDP \eqref{eq:pure_constr_dual}:
\begin{align}
a_{\min,d}  = &\sup_{a - m \in \cM(S(N))_d} m = \sup_{m} \inf_{L \in (\cM(S(N))_d)^{\vee} } (m + L(a -m)) \label{eq:pdfirst}\\
\leq & \inf_{L \in (\cM(S(N))_d)^{\vee} } \sup_{m}   \, (m + L(a -m)) \label{ineq:pdsecond}\\
= & \inf_{L \in (\cM(S(N))_d)^{\vee} } (L(a) +  \sup_{m} m (1 - L(1))) \label{eq:pdthird}\\
= & \inf_L \{L(f) \mid L \in (\cM(S(N))_d)^{\vee} \,, L(1) = 1 \} \label{eq:pdlast}\,,
\end{align}
The second equality in \eqref{eq:pdfirst} comes from the fact that the inner minimization problem gives minimal value 0 if and only if $a - m \in \cM(S(N))_d$.
The inequality in \eqref{ineq:pdsecond}  trivially holds.
The inner maximization problem in \eqref{eq:pdthird} is bounded with maximum value 0 if and only $L(1) = 1$.
Eventually, \eqref{eq:pdlast} is equivalent to SDP~\eqref{eq:pure_constr_primal} by Remark
\ref{rk:Hankelbij} and Lemma \ref{lemma:pureHankel}.
\end{proof}
Before proving that SDP \eqref{eq:pure_constr_dual} satisfies strong duality, we recall that an $\varepsilon$-neighborhood of 0 is the set $\mathcal{N}_{\varepsilon}$ defined for a given $\varepsilon > 0$ by:
\[
\mathcal{N}_{\varepsilon} := \bigcup_{k \in \N} \left\lbrace \underline{A} := (A_1,\dots,A_n) \in \Sbb_k^n : \varepsilon^2 - \sum_{i=1}^n A_i^2 \succeq 0  \right\rbrace \,.
\]
\begin{lemma}
\label{lemma:epsneighborhood}
If $f \in \fatT$ vanishes on an $\varepsilon$-neighborhood of 0, then $f = 0$.
\end{lemma}
\begin{proof}
We rely on the standard multilinearization trick and the fact that a trace polynomial $f\in \TX_d$ cannot be a trace identity  on $(d+1) \times (d+1)$ matrices, as a consequence of \cite[Theorem 4.5 (b)]{P76}.
Since $f$ vanishes on all $n$-tuples of $(d+1)\times (d+1)$ matrices  $\uA \in \mathcal{N}_{\varepsilon}$, one has $f = 0$.
\end{proof}
\begin{theorem}
\label{th:pure_hierarchy}
Let $S[N]$ be as in \eqref{eq:SN} and 
suppose that $\mathcal{D}_S$ contains an $\varepsilon$-neighbor\-hood of 0.
Then SDP \eqref{eq:pure_constr_dual} satisfies strong duality, i.e., there is no duality gap between SDP \eqref{eq:pure_constr_primal} and SDP \eqref{eq:pure_constr_dual}. 
\end{theorem}
\begin{proof}
The strong duality statement is proved as in \cite[Proposition 4.4]{cafuta2012constrained}. 
For this, we construct a linear map $L : \TX_{2d} \to \R$ which is a strictly feasible solution of SDP~\eqref{eq:pure_constr_primal}, namely $L(1) = 1$,  $L(a^2 s) > 0$ for all $s \in S$  and for all nonzero $a \in \skinnyT_{d - d_s}$, $L(a^2  (N^k - \Trace (x_j^{2k}) ) > 0$ for all $j$, $k\le d$ and for all nonzero $a \in \skinnyT_{d - k}$, and $L (\Trace(f f^\star)) > 0$ for all nonzero $f \in \TX_d$.
Let us pick $m > d$ and  consider the set $\mathcal{U}$ of $m \times m$ matrices from $\mathcal{D}_{S[N]}$ with rational entries, written as
\[
\mathcal{U} = \{\uA^{(k)} \mid k \in \N,\uA^{(k)}  \in  \mathcal{D}_{S[N]}^m \}\,.
\]
Note that $\mathcal{U}$ contains a dense subset of $m\times m$ matrices in $\cN_\varepsilon$.
Let us define
\[
L := \sum_{k=1}^{\infty} 2^{-k} \frac{L_{\uA^{(k)}}}{\|L_{\uA^{(k)}}\|} \,,
\]
with
\[
L_{\uA} : \TX_{2d} \to \R \,, \quad f \mapsto \Trace f(\uA) \,,
\]
for all $\uA \in \mathcal{U}$.
This functional $L$ is obviously linear and unital.
One has $L(\Trace (f f^\star)) \geq 0$ for all $f \in \TX_d$. Now let us assume that $L(\Trace (f f^\star)) = 0$ for some $f \in \TX_d$.
This implies that for all $k \in \N$ one has $\Trace (f(\uA^{(k)}) f^\star(\uA^{(k)})) = 0$, thus $f(\uA^{(k)}) f^\star(\uA^{(k)}) = 0$, which in turn implies that $f(\uA^{(k)}) = 0$.
By density of $\mathcal{U}$ in $\cN_\varepsilon\cap \Sbb_m^n$, $f(\uA) = 0$ for all $\uA\in \cN_\varepsilon\cap \Sbb_m^n$. As $m$ was arbitrary, $f$ vanishes on $\cN_\varepsilon$.
By Lemma \ref{lemma:epsneighborhood}, one has $f = 0$.
The two other positivity conditions are proved in a similar fashion.
\end{proof}

\begin{corollary}
\label{cor:pure_cvg}
The hierarchy of SDP programs \eqref{eq:pure_constr_dual} provides a sequence of lower bounds $(a_{\min,d})_{d \geq d_{\min}}$  monotonically converging to $a_{\min}^{\emph{II}_1}$.
\end{corollary}
\begin{proof}
As $\cM(S(N))_d \subseteq \cM(S(N))_{d+1}$, one has $a_{\min,d} \leq a_{\min,d+1}$.
Furthermore, Theorem~\ref{thm:nocyc} implies that for each each $m \in \N$, there exists $d(m) \in \N$ such that $a - a_{\min}^{\twoone} + \frac{1}{m} \in \cM(S(N))_{d(m)}$.
Thus one has
\[
a_{\min}^{\twoone} - \frac{1}{m} \leq  a_{\min,d(m)} \,,
\]
which implies that
\[
\lim_{d \to \infty} a_{\min,d}  = a_{\min}^{\twoone}
 \,.\qedhere\]
\end{proof}

\if{
\begin{remark}
\label{rk:qiconverges}
Corollary \ref{cor:pure_cvg} implies that one can build a converging sequence of SDP programs similar to the hierarchy presented in \cite{pozas2019bounding}, relying on scalar extension of the moment matrices involved in the NPA hierarchy \cite{navascues2008convergent}.
Indeed, the scalar extension matrix $\tilde \Gamma$ from \cite{pozas2019bounding} is indexed by a set of extra operators, which corresponds exactly to the set of $\fatT$-words. 
The matrix $\tilde{\Gamma}$ satisfies the tracial Hankel property.
In addition, the set of noncommuting operators $\underline{x} = (A_i, B_j, C_k)$, satisfy causal constraints, which boil down to equality constraints involving polynomials in $\emph{\skinnyT}$, e.g., $\Trace(A_{i_1} A_{i_2} \cdots A_{i_m} C_{k_1} C_{k_2} \cdots C_{k_m}) - \Trace(A_{i_1} A_{i_2} \cdots A_{i_m}) \Trace(C_{k_1} C_{k_2} \cdots C_{k_m}) = 0$.
The resulting pure trace localizing matrices are constrained to be zero.
In addition, the operator constraints $\{1 - x_j^2 = 0 \mid j=1,\dots,n \}$ are also taken into account by performing moment substitutions within $\tilde \Gamma$.
Overall, by taking $S$ to be the set all polynomials involved in the causal equality constraints (together with their opposites) and taking $S' = S \cup \{\pm (1 - x_j^2) \mid j=1,\dots,n \}$, the constraints of SDP \eqref{eq:pure_constr_primal} are equivalent to the ones stated in \cite{pozas2019bounding}.
As a consequence of Corollary \ref{cor:pure_cvg}, the hierarchy from  \cite{pozas2019bounding} is thus proved to be converging.
\end{remark}
}\fi
\subsection{Finite-dimensional GNS representations and minimizer extraction}
\label{sec:gns}
The goal of this section is to derive an algorithm to extract minimizers of pure trace polynomial optimization problems. 
The forthcoming statements can be seen as ``pure trace'' variants of the results derived in the context of commutative polynomials \cite{curto1998flat}, eigenvalue optimization of noncommutative  polynomials \cite[Lemma 2.2]{McCullSOS} (see also \cite{pironio2010convergent},~\cite[Chapter 21]{anjos2011handbook} and~\cite[Theorem~1.69]{burgdorf16}), and trace optimization of noncommutative polynomials~\cite{nctrace}.

\begin{definition}
\label{def:flatextension}
Suppose $L : \emph{\skinnyT}_{2d + 2 \delta} \to \R$ is a tracial linear functional with restriction $\tilde{L} : \emph{\skinnyT}_{2 d} \to \R$. 
We associate to $L$ and $\tilde{L}$ the Hankel matrices $\M_{d+\delta}^{\fatT} (L)$ and $\M_d^{\fatT} (\tilde L)$ respectively, and get the block form
\[
\M_{d + \delta}^{\fatT} (L) = \begin{bmatrix}
\M_d^{\fatT} (\tilde L) & B \\[1mm]
B^T & C
\end{bmatrix} \,.
\]
We say that $L$ is \emph{$\delta$-flat} or that $L$ is a \emph{$\delta$-flat extension} of $\tilde{L}$, if $\M_{d+\delta}^{\fatT} (L)$ is flat over $\M_d^{\fatT} (\tilde{L})$, i.e., if $\rank \M_{d+\delta}^{\fatT} (L) = \rank \M_d^{\fatT} (\tilde{L}) $.
\end{definition}

\revise{Suppose $L$ is $\delta$-flat and let $r := \rank \M_{d}^{\fatT}(L) = \M_{d + \delta}^{\fatT}(L)$.
Since $\M_{d+\delta}^{\fatT}(L) \succeq 0$, we obtain the Gram matrix decomposition 
$  \M_{d + \delta}^{\fatT}(L) = [\langle \u , \w \rangle]_{u,w}$ with vectors $\u,\w \in \R^r$, where the labels $u,v$ are $\fatT$-words of degree at most $d + \delta$. 
Then, we define the following finite-dimensional Hilbert space  
\[
\cH := \text{span} \, \{\w \mid \deg w \leq d + \delta \} = \text{span} \, \{\w \mid \deg w \leq d \} ,
\]
where the equality is a consequence of the flatness assumption.
Afterwards, one can follow the steps performed in Appendix~\ref{app} (see \eqref{e:creation}) and consider, for each $p \in \fatT$, the multiplication operator $\hat \chi_p$ on $\cH$ and the $\star$-representation $\pi : \fatT \to \mathcal{B}(\cH)$ defined by $\pi(p) = \hat \chi_p$.
Let $\vb$ be the vector representing 1 in $\cH$; then $L(p) = \langle \pi(p) \vb, \vb \rangle$ for all $p \in \fatT$. 
In general, elements of $\pi(\skinnyT)$ are central in $\pi(\fatT)$; if they are actually scalar multiples of the identity on $\cH$, then $\pi$ is not just a $\star$-representation, but it respects trace in the sense that $\pi(f)=f(\pi(x_1),\dots,\pi(x_n))$ for every $f\in\fatT$. 
This fact applies to our SDP hierarchy as follows.

\begin{proposition}
\label{prop:pure_trace_flat}
Given $S \cup \{a \} \subseteq \emph{T}_{2 d}$,
 let $S[N]$ be as in \eqref{eq:SN}. 
Set $\delta := \max \{ \lceil \deg s /2 \rceil : s \in S[N] \}$.
Assume that $L$ is a $\delta$-flat optimal solution of SDP~\eqref{eq:pure_constr_primal}, and assume that $\pi(\skinnyT)=\R$, where $\pi : \fatT \to \mathcal{B}(\cH)$ is the $\star$-representation constructed above.
Then, one has
\begin{equation}
\label{eq:pure_trace_flat}
a_{\min,d + \delta} = L(a) = a_{\min}^{\II_1} \,.
\end{equation}
Moreover, there are finitely many $n$-tuples $\uA^{(j)}$ of symmetric matrices,
and positive scalars $\lambda_j$ with $\sum_j \lambda_j = 1$, such that $a_{\min}^{\II_1}=a(\bigoplus_j \uA^{(j)})$, where the tracial state is given by
\[
w\left(\bigoplus_j \uA^{(j)}\right)\mapsto
\sum_j \lambda_j \Trace (w(\uA^{(j)}))
\]
for $w\in\mx$.
\end{proposition}

}

\begin{proof}
For $i=1,\dots,n$ let $A_i = \hat \chi_{x_i}$ be the left multiplication by $x_i$ on $\cH$, i.e., for each $\fatT$-word $w \in \TX_{d}$, $A_i \w$ is the vector from $\cH$ corresponding to the label $x_i w$.
The operators $A_i$ are well-defined (thanks to the flatness assumption) and symmetric. After choosing an orthonormal basis of $\cH$ we can view $A_i$ as $r\times r$ symmetric matrices.
Let $\uA := (A_1,\dots,A_n)$, and let $\cA\subseteq \Mbb_r$ be the algebra generated by $A_1,\dots,A_n$. 
\revise{
Since $\pi(\skinnyT)=\R$, the map $\tau:\cA\to\R$ given by $q(\uA)\mapsto \pi(\Trace q)=L(q)$ for $q\in\RX$ is a well-defined faithful tracial state on $\cA$.
}
For each $s \in S$, one has $s(\uA) = \langle \pi(s) \vb , \vb \rangle = L(s) \geq 0$, where the last inequality follows from the fact that $\M_{d - d_s}^{\skinnyT}(s \, L) \succeq 0$ as $L$ is a feasible solution of SDP~\eqref{eq:pure_constr_primal}. 
Similarly, one has $N^{k} - \tau (A_j^{2k}) \geq 0 $, for all $j=1,\dots,n\,, k\le d$.
Therefore, $\uA \in \mathcal{D}_{S[N]}^{\trvN}$. 

Eventually, $ a_{\min,d + \delta} = L(a) \leq a_{\min}^{\twoone}$, where the first equality is the strong duality statement from Theorem \ref{th:pure_hierarchy}. In addition, one has $a_{\min}^{\twoone} \leq  a(\uA) = L(a)$, yielding the desired result \eqref{eq:pure_trace_flat}.

We get a tracial representation of the optimizer for $a_{\min}^{\twoone}$ by performing the Artin-Wedderburn block diagonalization on the algebra $\cA$. This step relies on the Wedderburn theorem \cite[Chapter 1]{Lam13}. By \cite[Proposition 1.68]{burgdorf16}, there are finitely many tuples of symmetric matrices $\uA^{(j)}$ and positive scalars $\lambda_j$ with $\sum_j\lambda_j=1$ such that
$$\tau(q(\uA))=\sum_j\lambda_j \Trace(q(\uA^{(j)}))$$
for all $q\in\RX$.
\end{proof}

\revise{
\begin{remark}
The condition $\pi(\skinnyT) = \R$ in Proposition \ref{prop:pure_trace_flat} in particular holds if $L$ is an extreme optimal solution of \eqref{eq:pure_constr_primal} (cf.~the last part of the proof in Appendix~\ref{app}).	
In practice modern SDP solvers rely on interior-point methods using the so-called ``self-dual embedding'' technique \cite[Chapter 5]{wolkowicz2012handbook}. 
Therefore, they will always converge towards an optimum solution of maximum rank, yielding an extreme linear functional; see \cite[\S4.4.1]{lasserre2008semidefinite} for more  details.
\end{remark}
}

\begin{remark}
Proposition \ref{prop:pure_trace_flat} guarantees that in \revise{the} presence of a flat extension, there is an optimizer for $a_{\min}^{\twoone}$ arising from a finite-dimensional tracial pair $(\cF,\tau)$; furthermore, the dimensions of $\uA^{(j)}$ and the scalars $\lambda_j$ explicitly determine $\cF$ and $\tau$, respectively. It is sensible to ask whether $a_{\min}=a_{\min}^{\twoone}$, that is, whether 
the optimum can be approximated arbitrarily well with a
 finite-dimensional factor, i.e., from $\cD_{S[N]}$. If there exist sequences of positive rational numbers $(\lambda_j^{(m)})_m$ such that $\sum_j\lambda_j^{(m)}=1$ for all $m\in\N$, $\lim_m \lambda_j^{(m)}=\lambda_j$ for all $j$, and
$\bigoplus_j \uA^{(j)} \in\cD_{S[N]}^\trvN$ whenever the tracial state is given by
\begin{equation}\label{e:rat}
w\left(\bigoplus_j \uA^{(j)}\right)\mapsto
\sum_j \lambda_j^{(n)} \Trace (w(\uA^{(j)})) \qquad \text{for } w\in\mx,
\end{equation}
then $a_{\min}=a_{\min}^{\twoone}$. Indeed, a finite-dimensional tracial pair with the rational-coefficient tracial state as in \eqref{e:rat} embeds into a finite-dimensional factor. However, in general $a_{\min}\neq a_{\min}^{\twoone}$ even if $\cD_S$ contains an $\varepsilon$-neighborhood of $0$ and $a_{\min}^{\twoone}$ admits a finite-dimensional optimizer; see the following example.
\end{remark}

\begin{example}
Fix $n=1$, i.e., $\emph{\skinnyT} = \R[\Trace(x_1^i)\mid i\in\N]$. For $k\in\N$ let
$$
s_k:=1+(\sqrt{2}+1)^2-\left(\Trace\left((x_1^2-2x_1)^2\right)+\left(\sqrt{2}-\Trace(x_1)\right)^2\right)\Trace(x_1^{2k}) \in\emph{\skinnyT}.
$$
Let $X$ be a symmetric matrix. Then $X^2\neq2X$ or $\Trace (X)\neq\sqrt{2}$. Furthermore, if $X$ is a contraction, then
$$0\preceq 2X-X^2\preceq I,\quad
|\sqrt{2}-\Trace(X)|\le \sqrt{2}+1,\quad
\Trace(X^{2k})\le1\ \text{ for all }k\in\N.$$
On the other hand, if $X$ is not a contraction, then there is $k\in\N$ such that
$$\Trace(X^{2k})>
\frac{1+(\sqrt{2}+1)^2}{\Trace\left((X^2-2X)^2\right)+\left(\sqrt{2}-\Trace(X)\right)^2}.
$$
Let $S=\{s_k\mid k\in\N \}$ and $a=-\Trace(x_1)$. Then $\cD_S=\cD_{1-x_1^2}$ by the above observations, and consequently $a_{\min}=-1$. On the other hand, consider the tracial pair $(\R^2,\tau)$ with $\tau(\xi_1,\xi_2)=\frac{1}{\sqrt{2}}\xi_1+(1-\frac{1}{\sqrt{2}})\xi_2$. Then $Y=(2,0)\in\R^2$ satisfies $Y^2=2Y$ and $\tau(Y)=\sqrt{2}$, so $Y\in \cD_S^{\emph{\trvN}}$. Therefore $a_{\min}^{\twoone}<a(Y)=-\sqrt{2}$.
\end{example}

The proof of Proposition \ref{prop:pure_trace_flat} gives the following procedure for minimizer extraction.
\begin{algorithm}
$\gns$
\label{algorithm:gns}
\begin{algorithmic}[1]
\Require an extreme $\delta$-flat linear $L : \emph{\skinnyT}_{2 d + 2 \delta}\to \R$  solution of~\eqref{eq:pure_constr_primal}.

\State Let us consider the set of $\fatT$-words $\{w_i\}$ of degree at most $ \leqslant d$, such that
$\mathscr{C}$, the matrix consisting of columns of $\M(L)$ indexed by the words $w_1,\dots,w_r$, has full rank. 
Assume $w_1 = 1$.
\State Let $\M(\hat L)$ be the principal submatrix of $\M(L)$ of columns and rows indexed by $w_1,\dots, w_r$.
\State Let $C$ be the Cholesky factor of $\M(\hat L)$, i.e., $C^T C = \M(\hat L)$.
\For{$i \in \{1,\dots, n\}$} 
\State Let $\mathscr{C}_i$ be the matrix consisting of columns of $\M(L)$ indexed by $x_i w_1,\dots,x_i w_r$.
\State Compute $\bar{A}_i$ as a solution of the system $\mathscr{C} \bar{A}_i = \mathscr{C}_i$.
\State Let $A_i = C \bar{A}_i C^{-1}$.
\EndFor
\State Compute $\vb = C \e_1$. \Comment{$e_1 = (1,0,\dots,0)$}
\State Let $\cA \subseteq \Mbb_r$ be the algebra generated by $A_1,\dots,A_n$.
Compute an orthogonal matrix  $Q$ performing the simultaneous block-diagonalization of $A_1,\dots,A_n$ by \cite[Algorithm 4.1]{MKKK10}.
\Comment{$Q^T \cA Q = \{ \Diag (B^{(1)},\dots, B^{(k)}) \mid B^{(i)} \in \cA_i \}$ where $\cA_1,\dots, \cA_k$ are simple $\star$-algebras over $\R$ }
\State Compute $Q^T A_i Q = \Diag (A_i^{(1)},\dots, A_i^{(k)})$ for each $i = 1, \dots, n$, 
	and $Q^T \vb = ((\vb^1)^T,\dots,(\vb^k)^T)^T$.
\State Compute $\lambda_j = \|\vb^j \|$, and $\uA^{(j)} = ( A_1^{(j)},\dots, A_n^{(j)}) $, for all $j = 1,\dots,k$.
\Ensure  $(\underline A^{(1)}, \dots, \underline A^{(k)})$  and $(\lambda_1,\dots, \lambda_k)$.
\end{algorithmic}
\end{algorithm}
The correctness of the procedure $\gns$ follows from the proof of Proposition~\ref{prop:pure_trace_flat}.
\begin{corollary}
\label{coro:pure_gns_sound}
The procedure $\gns$ described in Algorithm~\ref{algorithm:gns} is sound and returns the $n$-tuples $\underline{A}^{(j)}$ and $\lambda_j$ from Proposition~\ref{prop:pure_trace_flat}.
\end{corollary}

\begin{remark}
\label{rk:thin}
Note that when flatness occurs, Proposition \ref{prop:pure_trace_flat} guarantees convergence (actually stabilization) of our SDP hierarchy even if there is no $\varepsilon$-neighborhood of $0$ in the feasible set. Moreover, while a flat extension is evasive from a numerical point of view, an ``almost'' flat extension, which is a much more viable output of an SDP solver, is likely sufficient \cite{robust}.
\end{remark}

\subsection{SDP hierarchy for trace polynomial optimization}
\label{sec:constrained}

Here we describe the reduction from the general trace setting to the pure trace setting.

Let $S\subseteq\Sym\T$ and $N>0$. 
Denote 
\begin{align}
\label{eq:tildeD}
\widetilde{S}=\{\Trace(fsf^\star)\mid s\in S,\ f\in\T\}\subseteq \emph{\skinnyT}.
\end{align}
\begin{proposition}\label{prop:slacking}
Let $S\subseteq\Sym\T$, $N>0$, and let $\widetilde{S}$ be as in \eqref{eq:tildeD}.
Then $\mathcal{D}_{\widetilde{S}[N]}^{\cF,\tau}=\mathcal{D}_{S[N]}^{\cF,\tau}$ for any tracial pair $(\cF,\tau)$. Furthermore, the following are equivalent for $a\in\skinnyT$:
\begin{enumerate}[\rm (i)]
\item\label{it:k1} $a(\uX)\ge 0$ for all $\uX \in \mathcal{D}_{S[N]}^{\trvN}$; \vspace{0.25ex}
\item \label{it:k2} $a(\uX)\ge 0$ for all 
$\uX \in \mathcal{D}_{S[N]}^{\II_1}$; \vspace{0.25ex}
\item\label{it:k3} $a+\varepsilon \in \cM(\widetilde{S}(N))$ for all $\varepsilon>0$.
\end{enumerate}
\end{proposition}

\begin{proof}
The equality $\mathcal{D}_{\widetilde{S}[N]}^{\cF,\tau}=\mathcal{D}_{S[N]}^{\cF,\tau}$ follows from Proposition \ref{p:poly}. Consequently, \ref{it:k1}$\Leftrightarrow$\ref{it:k2}$\Leftrightarrow$\ref{it:k3} holds by Theorem \ref{thm:nocyc}.
\end{proof}

For all $d\in\N$, one has 
$$\cM(\widetilde{S}(N))_d =\left\{
\sum_{i=1}^{K} a_i^2s_i\mid K\in\N,\ a_i\in\skinnyT,\ s_i \in \widetilde{S}(N),\ \deg(a_i^2s_i)\le 2d
\right\}.$$
Therefore, elements of $\cM(\widetilde{S}(N))_d$ corresponds to sums of elements of the form
\begin{align}
\label{eq:elementstildeM}
\Trace (f_1 \, s \, f_1^\star) \,, \quad a^2 \big(N^{k} - \Trace (x_j^{2k})\big) \,, \quad \Trace (f_2 f_2^\star) \,,
\end{align}
which are of degree at most $2d$, for $f_i \in \fatT$, $a \in \skinnyT$, $s \in S$, $1\le j\le n$, $k\in\N$.

As in Section~\ref{sec:pure_constrained},  given $a \in \skinnyT$, 
one can under-approximate $a_{\min}^{\twoone}$ via the following hierarchy of SDP programs, indexed by $d \geq d_{\min}$: 
\begin{align}
\label{eq:tilde_constr_dual}
\widetilde a_{\min,d} = \sup \{ m \mid a - m \in \cM(\widetilde{S}(N))_d \} \,.
\end{align}
The dual of \eqref{eq:tilde_constr_dual} is obtained by replacing the pure trace localizing matrix constraints in SDP~\eqref{eq:pure_constr_primal} by trace localizing matrix constraints associated to each $s \in S$:
\begin{equation}
\label{eq:tilde_constr_primal}
\begin{aligned}
\inf_{\substack{L : \text{T}_{2 d} \to \R \\ L \text{ linear}}} \quad  & L(a)  \\	
\emph{s.t.} 
\quad & (\M_d^\fatT(L))_{u,v} = (\M_d^\fatT(L))_{w,z}  \,, \quad \text{whenever } \Trace(u^\star v) =  \Trace( w^\star z) \,, \\
\quad & (\M_d^\fatT(L))_{1,1} = 1 \,, \\
\quad & \M_d^{\fatT}(L) \succeq 0 \,, \\
\quad & \M_{d - d_s}^{\fatT}(s \, L) \succeq 0  \,,  \quad \text{for all }  s \in S  \,,\\
\quad & \M_{d - k}^{\text{T}}( (N^{k} - \Trace (x_j^{2k}) ) \, L) \succeq 0  \,,  \quad \text{for all }  j=1,\dots,n\,, k\le d  \,.
\end{aligned}
\end{equation}
As in Theorem~\ref{th:pure_hierarchy}, one can prove that if $\mathcal{D}_S$ contains an $\varepsilon$-neighborhood of 0, then there is no duality gap between SDP \eqref{eq:tilde_constr_primal} and SDP \eqref{eq:tilde_constr_dual}. 
In addition, the hierarchy of SDP programs \eqref{eq:tilde_constr_dual} provides a sequence of lower bounds monotonically converging to $a_{\min}^{\II_1}$.

\revise{\begin{remark}
There are several variations of \eqref{eq:tilde_constr_primal} that lead to $a_{\min}^{\II_1}$.
For example, one can replace the $d$ matrix inequalities 
$\M_{d - k}^{\text{T}}( (N^{k} - \Trace (x_j^{2k}) ) \, L) \succeq 0$ (for a fixed $j$) with a single matrix inequality 
$\M_{d - 1}^{\fatT}( (N - x_j^2 ) \, L) \succeq 0$
since
$$N^{k} - \Trace (x_j^{2k}) = \Trace\left(
\sum_{i=1}^{k-1} (\sqrt{N}^{k-1-i}x_j^i)(N-x_j^2)(\sqrt{N}^{k-1-i}x_j^i)
\right).$$
While this modification produces a somewhat simpler-looking SDP, note that the new matrix constraint has size $\bsigma^{\fatT} (n,d - 1)$, while the combined size of the replaced constraints equals
$\sum_{k=1}^d\bsigma^{\skinnyT} (n,k)$, which is less than $\bsigma^{\fatT} (n,d - 1)$.

In practice, when given a concrete set of constraints $S$, one should attempt to simplify \eqref{eq:tilde_constr_primal} before solving it, since its most general form can contain superfluous inequalities with respect to $S$.
\end{remark}}

Finally, the next result provides an alternative characterization of (not necessarily pure) trace polynomials positive on tracial semialgebraic sets (cf. Corollary \ref{cor:linslackpsatz}).
\begin{proposition}\label{prop:slackingfat}
Let $S\subseteq\Sym\T$, $N>0$, and let $\widetilde{S}$ be as in \eqref{eq:tildeD}.
For $a \in \Sym\T$, the following are equivalent:
\begin{enumerate}[\rm (i)]
	\item\label{it:l1} $a(\uX)\succeq 0$ for all 
	$\uX \in \mathcal{D}_{S[N]}^{\trvN}$;\vspace{0.25ex}
	\item\label{it:l2} $a(\uX)\succeq 0$ for all 
	$\uX \in \mathcal{D}_{S[N]}^{\II_1}$;\vspace{0.25ex}
	\item\label{it:l3} for every $\varepsilon>0$, there exist sums of (two) squares $s_1,s_2\in\R[t]$ such that
	\begin{equation}\label{e:woy1}
		a=s_1(a)-s_2(a),\qquad \varepsilon-\Trace(s_2(a))\in \cM(\widetilde{S}(N));
	\end{equation}
	\item\label{it:l4} for every $\varepsilon>0$, there exist sums of (two) squares $s_1,s_2\in\R[t]$ and $q\in \cM(\widetilde{S}(N))$ such that
	\begin{equation}\label{e:wy1}
	\Trace(ay)+\varepsilon =\Trace(s_1(a)y+s_2(a)(1-y))+q
	\end{equation}
	where $y$ is an auxiliary symmetric free variable.
\end{enumerate}
\end{proposition}

\begin{proof}
\ref{it:l1}$\Leftrightarrow$\ref{it:l2} holds by Proposition \ref{p:embed}, and \ref{it:l3}$\Rightarrow$\ref{it:l4} follows by taking $q=\varepsilon-\Trace(s_2(a))\in \cM(\widetilde{S}(N))$. Furthermore, \ref{it:l4}$\Rightarrow$\ref{it:l1} holds by Propositions \ref{p:poly} and \ref{prop:slacking}.

\ref{it:l1}$\Rightarrow$\ref{it:l3} There is $N'>0$, dependent on $N$ and $a$, such that $N'-a\succeq0$ on $\mathcal{D}_{S[N]}^{\trvN}$. After rescaling $a$ we can without loss of generality assume that $1-a\succeq0$ on $\mathcal{D}_{S[N]}^{\trvN}$. Suppose that \ref{it:l1} holds. Given an arbitrary $\varepsilon>0$ let $s_1,s_2\in\R[t]$ be sums of squares as in Lemma \ref{l:univar}. Then there are sums of squares $s_3,s_4,s_5\in\R[t]$ such that
$$s_1-s_2=t,\qquad \tfrac{\varepsilon}{2}-s_2 =s_3+s_4t+s_5(1-t).$$
By Propositions \ref{p:poly} and \ref{prop:slacking} we have
$$\Trace(s_3(a)+s_4(a)a+s_5(a)(1-a))+\tfrac{\varepsilon}{2}\in \cM(\widetilde{S}(N)).$$
Therefore
\[a=s_1(a)-s_2(a),\qquad \varepsilon-\Trace(s_2(a))\in \cM(\widetilde{S}(N)). \qedhere\]

\end{proof}
Note that Proposition \ref{prop:slackingfat} allows one to certify that a given trace polynomial is positive semidefinite on a tracial semialgebraic set.
Constructing a hierarchy of SDP programs converging to the minimal eigenvalue of trace polynomials is postponed for future work. As Example \ref{ex:unclean} indicates, it cannot be simply derived from our scheme for the pure trace polynomial objective function; namely, the norm of an operator cannot be uniformly estimated with traces in a dimension-free way.

\revise{
\section{Examples and applications}
\label{sec:new}

In this section we present some experimental results indicating the strength and computational aspects of the SDP hierarchy in Section \ref{sec:hierarchy}.
First we give a toy example of optimizing a pure trace polynomial over all projections in tracial von Neumann algebras (Section \ref{sec:toy}).
Next we describe how our algorithms can be used for finding upper bounds on quantum violations of polynomial Bell inequalities in quantum information theory (Section \ref{sec:bell}).
While the SDPs were solved using SeDuMi in Matlab, the sparse input matrices were constructed with Mathematica.
	
\subsection{A toy example}
\label{sec:toy}

\def\rr{{\rm r}}

Consider the optimization problem
\begin{equation}
\label{eq:toy}
\begin{aligned}
\inf \quad  & \tau(X_1X_2X_3)+\tau(X_1X_2)\tau(X_3)  \\	
\emph{s.t.} 
\quad & X_j^2=X_j^*=X_j \quad \text{for}\ j=1,2,3.
\end{aligned}
\end{equation}
over triples $(X_1,X_2,X_3)$ of operators in tracial pairs $(\cF,\tau)$.
Note that if $\cF$ is a commutative von Neumann algebra with a tracial state $\tau$ and $X_1,X_2,X_3\in\cF$ are projections, then $\tau(X_1X_2X_3),\tau(X_1X_2),\tau(X_3)\ge0$. Hence if \eqref{eq:toy} were restricted only to commutative von Neumann algebras, the solution would be $0$. On the other hand, the projections
$$X_1=\begin{pmatrix}
1 & 0 \\ 0 & 0 
\end{pmatrix},\quad
X_2=\begin{pmatrix}
\tfrac{1}{16} & \tfrac{\sqrt{15}}{16} \\[1mm] \tfrac{\sqrt{15}}{16} & \tfrac{15}{16} 
\end{pmatrix},\quad
X_3=\begin{pmatrix}
\tfrac{3}{8} & -\tfrac{\sqrt{15}}{8} \\[1mm] -\tfrac{\sqrt{15}}{8} & \tfrac{5}{8} 
\end{pmatrix}
$$
give
$$\Trace(X_1X_2X_3)+\Trace(X_1X_2)\Trace(X_3)=-\frac{1}{32}.$$
We next show that $-\frac{1}{32}$ is actually the solution of \eqref{eq:toy}.

Let $n=3$, $a=\Trace(x_1x_2x_3)+\Trace(x_1x_2)\Trace(x_3)$ and $S=\{x_j^2-x_j,x_j-x_j^2\colon j=1,2,3\}$.
By Section \ref{sec:constrained}, the solution of \eqref{eq:toy} equals $\lim_{n\to\infty}\check{a}_{\min,d}$, where $\check{a}_{\min,d}$ is the solution of \eqref{eq:pure_constr_primal} for $d\ge2$. In this particular example, the constraints can be used to vastly simplify \eqref{eq:pure_constr_primal}. Namely, it suffices to consider only tracial words without consecutive repetitions of $x_j$; furthermore, the last two lines in \eqref{eq:pure_constr_primal} are then superfluous. To state this concretely, let us introduce some auxiliary notation.

A $\fatT$-word is {\it square--reduced} if no proper powers of $x_1,x_2,x_3$ appear in it. To each $\fatT$-word $w$ we can assign the square--reduced $\fatT$-word $\rr(w)$ by repeatedly replacing $x_j^2$ with $x_j$. Let $\W^\rr_d$ be the vector of all square--reduced $\fatT$-words of tracial degree at most $d$, and let $R_d$ be the span of entries of $\W^\rr_d$. 
Given a linear functional $L:R_{2d}\to\R$, the {\it square--reduced tracial Hankel matrix} $\M_d^\rr(L)$ 
is indexed by $\W^\rr_d$ and $(\M_d^\rr(L))_{u,v}=L(\Trace(\rr(u^\star v)))$. Then $\check{a}_{\min,d}$ is the solution of the SDP
\begin{equation}
\label{eq:toy_sdp}
\begin{aligned}
\inf_{\substack{L : R_{2 d} \to \R \\ L \text{ linear}}} \quad  & L(a)  \\	
\emph{s.t.} 
\quad & (\M_d^\rr(L))_{u,v} = (\M_d^\rr(L))_{w,z}  \,, \quad \text{whenever } \Trace(u^\star v) =  \Trace( w^\star z) \,, \\
\quad & (\M_d^\rr(L))_{1,1} = 1 \,, \\
\quad & \M_d^\rr(L) \succeq 0. 
\end{aligned}
\end{equation}

We start with $d=2$. The matrix $\M_2^\rr(L)$ is indexed by reduced tracial words
\begin{align*}
& 1,x_1,x_2,x_3,\Trace(x_1),\Trace(x_2),\Trace(x_3),\\
& x_1x_2,x_2x_1,x_1x_3,x_3x_1,x_2x_3,x_3x_2,\Trace(x_1x_2),\Trace(x_1x_3),\Trace(x_2x_3), \\
& \Trace(x_1)x_1,\Trace(x_1)x_2,\Trace(x_1)x_3,
\Trace(x_2)x_1,\Trace(x_2)x_2,\Trace(x_2)x_3,
\Trace(x_3)x_1,\Trace(x_3)x_2,\Trace(x_3)x_3, \\
& \Trace(x_1)^2,\Trace(x_2)^2,\Trace(x_3)^2,
\Trace(x_1)\Trace(x_2),\Trace(x_1)\Trace(x_3),\Trace(x_2)\Trace(x_3).
\end{align*}
The SDP \eqref{eq:toy_sdp} minimizes over $31\times 31$ positive semidefinite matrices subject to $881$ linear equations in their entries. By solving it we get $\check{a}_{\min,2}=-0.0467$.

In the next step we have $d=3$, and $\M_3^\rr(L)$ is a $108\times 108$ matrix with $11270$ linear relations. Now the solution of \eqref{eq:toy_sdp} is $\check{a}_{\min,3}=-0.0312$, which up to floating point precision agrees with $-\frac{1}{32}$. Since $\check{a}_{\min,3}$ is a lower bound for the solution of \eqref{eq:toy} and is attained by the $2\times 2$ projections above, we conclude that $-\frac{1}{32}$ is the solution of \eqref{eq:toy}.

Similar optimization problem can be used for detecting quantum entanglement \cite{horodecki,bardet}. Namely, 
trace polynomials are in correspondence with invariant operators \cite{huber2020positive}, and pure trace polynomials, positive subject to certain constraints, then relate to invariant operators positive on separable Werner states \cite{werner,eggeling}.
Pure trace polynomial optimization can be thus used to efficiently produce entanglement witnesses for Werner states, which is a work in preparation.

\subsection{Polynomial Bell inequalities}
\label{sec:bell}

\def\cov{\operatorname{cov}}

In this section we connect trace polynomial optimization to violations of nonlinear Bell inequalities, outline a few examples and prove the optimal bound on maximal violation of the covariance Bell inequality considered in \cite{PHBB}, see Example \ref{exa2} below.

As a prelude, the classical Bell inequality states that
\begin{equation}\label{e:bell1}
\psi^* (A_1\otimes B_1+A_1\otimes B_2+A_2\otimes B_1-A_2\otimes B_2)\psi
\end{equation}
is at most $2$
for all separable states $\psi \in \C^k\otimes \C^k$ and 
$A_j,B_j \in \Mbb_k$ satisfying $A_j^*=A_j$, $A_j^2=I$, $B_j^*=B_j$, $B_j^2=I$.
Tsirelson's bound implies that \eqref{e:bell1} is at most $2\sqrt{2}$ when arbitrary states are allowed. Moreover, the maximal value $2\sqrt{2}$ is attained when $k=2$ and 
$\psi= \frac{1}{\sqrt{2}}(e_1\otimes e_1+e_2\otimes e_2)$.
In general, if $\psi_k$ is the generalized Bell state,
$$\psi_k= \frac{1}{\sqrt{k}}\sum_{j=1}^k e_j\otimes e_j \in \R^k\otimes \R^k,$$
which is a maximally entangled bipartite state on $\C^k\otimes \C^k$, unique up to bipartite unitary equivalence,
then
\begin{equation}\label{e:ghz}
\psi_k^*(X\otimes Y)\psi_k = \Trace(XY)
\end{equation}
for all $X,Y\in \Sbb_k$. Therefore Tsirelson's bound for \eqref{e:bell1} on maximally entangled states can be recovered as a pure trace polynomial optimization problem
$$\sup \ \Trace(x_1y_1)+\Trace(x_1y_2)+\Trace(x_2y_1)-\Trace(x_2y_2)\ \text{ s.t. }\ x_j^2=1,y_j^2=1.$$
Note that for finding the maximal violation for arbitrary $k\in\N$, there is no loss of generality if only symmetric matrices are considered instead of hermitian ones, since
$$\psi_k^*(Z\otimes W)\psi_k = \Trace(Z\overline{W})
=\Trace\left(
\begin{pmatrix}
\frac{Z+\overline{Z}}{2} & \frac{\overline{Z}-Z}{2i} \\
\frac{Z-\overline{Z}}{2i} & \frac{Z+\overline{Z}}{2} 
\end{pmatrix}
\begin{pmatrix}
\frac{W+\overline{W}}{2} & \frac{W-\overline{W}}{2i} \\
\frac{\overline{W}-W}{2i} & \frac{W+\overline{W}}{2} 
\end{pmatrix}
\right)$$
for all hermitian $k\times k$ matrices $Z,W$.

Upper bounds on quantum violations of linear Bell inequalities can be found using the NPA hierarchy \cite{navascues2008convergent} for eigenvalue optimization of noncommutative polynomials; for example, one can get Tsirelson's bound on violations of \eqref{e:bell1} by eigenvalue-optimizing $a_1b_1+a_1b_2+a_2b_1-a_2b_2$ subject to $a_j^2=b_j^2=1$ and $[a_i,b_j]=0$.

On the other hand, bilocal models \cite{BRGP,Chaves}, covariance of quantum correlations \cite{PHBB} and detection of partial separability \cite{Uffink} lead to more general {\em polynomial} Bell inequalities. While linear Bell inequalities are linear in expectation values of (products of) observables, polynomial Bell inequalities contain multivariate polynomials in expectation values of (products of) observables. For this reason, noncommutative polynomial optimization is not suitable for studying violations of nonlinear Bell inequalities. 
In contrast, trace polynomial optimization gives upper bounds on violations of polynomial Bell inequalities, at least for certain families of states, e.g. the maximally entangled bipartite states via \eqref{e:ghz}. We demonstrate this with the following examples.

\subsubsection{Example}

Consider a simple quadratic Bell inequality
\begin{equation}\label{e:bell2}
\big(\psi^* (A_1\otimes B_2+A_2\otimes B_1)\psi\big)^2+
\big(\psi^* (A_2\otimes B_1-A_2\otimes B_2)\psi\big)^2
\le4
\end{equation}
given in \cite{Uffink}, where it is shown that \eqref{e:bell2} holds for all separable states $\psi$, and for all 2-dimensional states (i.e. all states when $k=2$). In \cite{NKI}, \eqref{e:bell2} is shown to hold for arbitrary states, meaning it admits no quantum violations. An alternative automatized proof of \eqref{e:bell2} for maximally entangled states of arbitrary dimension can be obtained by solving the optimization problem
\begin{equation}\label{e:bell3}
\sup\ 
\left(\Trace(x_1y_2+x_2y_1)\right)^2+
\left(\Trace(x_1y_1-x_2y_2)\right)^2 
\ \text{ s.t. }\ x_j^2=1,y_j^2=1 \text{ for }j=1,2
\end{equation}
Let $S=\{\pm(1-x_j^2),\pm(1-y_j^2)\colon j=1,2 \}$.
The relaxation of \eqref{e:bell3} with $d=2$ as in Section \ref{sec:constrained},
\begin{equation}\label{e:bell31}
\inf\ \mu \ \text{ s.t. }\ 
\mu-\left(\Trace(x_1y_2+x_2y_1)\right)^2-
\left(\Trace(x_1y_1-x_2y_2)\right)^2 \in \mathcal{M}(S(1))_d
\end{equation}
outputs $4$, which coincides with the classical value in \eqref{e:bell2}. The concrete implementation of \eqref{e:bell31} encodes the relations $x_j^2=y_j^2=1$ directly in the SDP using reduced words analogously as in the toy Example \ref{sec:toy} (where projections were considered).

\subsubsection{Example}\label{exa2}

Another class of polynomial Bell inequalities arises from covariances of quantum correlations. Let
$$\cov_\psi(X,Y) = 
\psi^*(X\otimes Y)\psi-\psi^*(X\otimes I)\psi\cdot\psi^*(I\otimes Y)\psi$$
In \cite{PHBB} it is shown that while
\begin{equation}\label{e:bell4}
\begin{split}
&\cov_\psi(A_1,B_1)+\cov_\psi(A_1,B_2)+\cov_\psi(A_1,B_3) \\
+&\cov_\psi(A_2,B_1)+\cov_\psi(A_2,B_2)-\cov_\psi(A_2,B_3)\\
+&\cov_\psi(A_3,B_1)-\cov_\psi(A_3,B_2)
\end{split}
\end{equation}
is at most $\frac92$ for separable states $\psi$, it attains the value 5 with the Bell state $\psi_2$. The authors also performed extensive numerical search within entangled states for local dimensions $k\le 5$, but no higher value of \eqref{e:bell4} was found. They leave it as an open question whether higher dimensional entangled states could lead to larger violations \cite[Appendix D.1(b)]{PHBB}.

Let
\begin{align*}
a=&\Trace(x_1y_1)-\Trace(x_1)\Trace(y_1)+\Trace(x_1y_2)-\Trace(x_1)\Trace(y_2)+\Trace(x_1y_3)-\Trace(x_1)\Trace(y_3) \\
&+\Trace(x_2y_1)-\Trace(x_2)\Trace(y_1)+\Trace(x_2y_2)-\Trace(x_2)\Trace(y_2)-\Trace(x_2y_3)+\Trace(x_2)\Trace(y_3) \\
&+\Trace(x_3y_1)-\Trace(x_3)\Trace(y_1)-\Trace(x_3y_2)+\Trace(x_3)\Trace(y_2)\,.
\end{align*}
The relaxation of
\begin{equation}\label{e:bell5}
\sup\ a\ \text{ s.t. }\  x_j^2=1,y_j^2=1 \text{ for }j=1,2,3
\end{equation}
with $d=2$ returns 5. Therefore the value of \eqref{e:bell4} is at most 5 for every maximally entangled state, regardless of the local dimension $k$.

\subsubsection{Example}

A family of eight quadratic Bell inequalities (arising from linear ones via elimination) corresponding to a bilocal model for three parties $A,B,C$ is given in \cite{Chaves}:
\begin{equation}\label{e:bell6}
-\frac18 (J_1\pm J_2)^2-(\pm J_1\pm J_2+2) \le 0\,,
\end{equation}
where
\begin{align*}
J_1&=\sum_{i,j=1}^2 {\psi'}^*(A_i\otimes B'_1)\psi'\cdot {\psi''}^*(B''_1\otimes C_j)\psi''\,,\\
J_2&=\sum_{i,j=1}^2 (-1)^{i+j}{\psi'}^*(A_i\otimes B'_2)\psi'\cdot {\psi''}^*(B''_2\otimes C_j)\psi''\,,
\end{align*}
and $A_i,B'_j,B''_j,C_j$ are projections. Inequalities \eqref{e:bell6} are valid for every pair of separable states $\psi',\psi''$, and are equivalent to $\sqrt{|J_1|}+\sqrt{|J_2|}\le 1$ as derived in \cite{BRGP}. An upper bound on violations of \eqref{e:bell6} for maximally entangled shared states $\psi',\psi''$ is given by 
\begin{equation}\label{e:bell7}
\sup\ a\ \text{ s.t. }\  x_j^2=x_j,{y'}_j^2=y'_j,{y''}_j^2=y''_j,z_j^2=z_j \text{ for }j=1,2
\end{equation}
where
\begin{align*}
a=& -\frac18 \left(\sum_{i,j} \Trace(x_iy'_1)\Trace(y''_1z_j)\pm \sum_{i,j} (-1)^{i+j}\Trace(x_iy'_2)\Trace(y''_2z_j)\right)^2 \\
&\pm \sum_{i,j} \Trace(x_iy'_1)\Trace(y''_1z_j)
\pm \sum_{i,j} (-1)^{i+j}\Trace(x_iy'_2)\Trace(y''_2z_j)-2\,.
\end{align*}
While \eqref{e:bell7} fits in the trace polynomial optimization scheme presented in this paper, SDPs arising from \eqref{e:bell7} are very large because $a$ is a pure trace polynomial of degree $8$ in $8$ variables. For computing upper bounds on \eqref{e:bell7} to become viable, the sizes of SDPs will need to be reduced using sparsity and symmetry techniques, which we plan to develop later.

}

\section{Conclusion and perspectives}
\label{sec:concl}
We have derived several novel Positivstellens\"atze for  trace  polynomials positive on tracial semialgebraic sets.
Our tracial analog of Putinar's Positivstellensatz yields a converging hierarchy of semidefinite relaxations for optimizing pure trace polynomials under pure trace polynomial inequality constraints. 
We also provide an algorithm to extract minimizers of such  problems, thanks to a finite-dimensional  Gelfand-Naimark-Segal construction.

A topic of future research is to derive a  hierarchy of primal-dual SDP programs converging to the minimal eigenvalue of a trace polynomial under trace polynomial inequality constraints.  
A short-term research investigation track is to rely on this hierarchy to tackle trace polynomial problems arising from quantum information theory. 
Sharing the same \revise{computational} drawbacks as the classical Lasserre's hierarchy, our tracial framework will be limited to optimization problems involving a modest number of variables.
To overcome this scalability issue, we intend to focus on exploiting structural properties of the input data. 
One possibility is to extend the framework from \cite{klep2019sparse} to optimization problems involving sparse trace polynomials or the one from \cite{Riener13Symmetries} to problems involving symmetries.

\revise{
\subsection*{Acknowledgments}
The authors thank anonymous referees for their valuable comments and suggestions, which greatly improved presentation of the paper and demonstration of the main results. 
}

\begin{appendices}
	
\section{Alternative proof of Theorem \ref{thm:psatz}}\label{app}

\begin{proof}[Proof of \ref{it:j2}$\Rightarrow$\ref{it:j1}]
Assume $a+\varepsilon \notin \cQ$ for some $\varepsilon>0$. Let $U=\{p\in\Sym\T\mid \Trace(p)=0 \}$. Then $\cQ+U$ is a convex cone in $\Sym\T$. Since $a$ is a pure trace polynomial, we have $a+\varepsilon\notin \cQ+U$. Since $\cQ$ is archimedean, for every $p\in\Sym\T$ there exists $\delta>0$ such that $1\pm\delta p\in \cQ$, which in terms of \cite[Definition III.1.6]{Bar} means that $1$ is an algebraic interior point of the cone $\cQ+U$ in $\Sym\T$. By the Eidelheit-Kakutani separation theorem \cite[Corollary III.1.7]{Bar} there is a nonzero $\R$-linear functional $L_0:\Sym\T\to\R$ satisfying $L_0(\cQ+U)\subseteq\R_{\ge0}$ and $L_0(a+\varepsilon)\le 0$. In particular, $L_0(U)=\{0\}$. Moreover, $L_0(1)>0$ because $\cQ$ is archimedean, so after rescaling we can assume $L_0(1)=1$. Let $L:\T\to\R$ be the symmetric extension of $L_0$, i.e., $L(p)=\frac12 L_0(p+p^\star)$ for $p\in\T$. Note that
\begin{equation}\label{e:tr}
L(p)=L(\Trace(p))
\end{equation}
for all $p\in\T$, and in particular $L(pq)=L(qp)$ for all $p,q\in\T$.

Now consider the set $\cC$ of all symmetric linear functionals $L':\T\to \R$ satisfying $L'(\cQ+U)\subseteq\R_{\ge0}$ and $L'(1)=1$. This set is nonempty because $L\in \cC$. Endow $\T$ with the norm
$$\|p\|=\max\left\{\|p(\underline{X})\|\mid n\in\N, \underline{X}\in \Sbb_n^k, \|X_j\|\le1\right\}.$$
This is indeed a norm because no nonzero trace polynomial vanishes on matrices of all finite sizes. By the Banach-Alaoglu theorem \cite[Theorem III.2.9]{Bar}, the convex set $\cC$ is weak*-compact. Thus by the Krein-Milman theorem \cite[Theorem III.4.1]{Bar} we may assume that our separating functional $L$ is an extreme point of $\cC$.

On $\T$ we define a semi-scalar product $\langle p, q\rangle=L(pq^\star)$. By the Cauchy-Schwarz inequality for semi-scalar products,
$$\cN=\left\{q \in \T \mid L(qq^\star)=0\right\}$$
is a linear subspace of $\T$. Let $p,q\in\T$. Since $\cQ$ is archimedean, there exists $\delta>0$ such that $1-\delta p p^\star\in \cQ$ and therefore
\begin{equation}\label{e:ineq}
0\le L(q(1-\delta p p^\star)q^\star)=L(q q^\star)-\delta L(qpp^\star q^\star)\le L(q q^\star).
\end{equation}
In particular, $q\in\cN$ implies $qp\in\cN$, so $\cN$ is a left ideal. Furthermore, $L(\cN)=\{0\}$: if $L(qq^\star)=0$, then for every $\delta>0$,
$$0\le L((\delta\pm q)(\delta\pm q)^\star)=\delta(\delta\pm 2 L(q))$$
and hence $L(q)=0$. Let $\overline p=p+\cN$ denote the residue class of $p\in\T$ in $\T/\cN$. Because $\cN$ is a left ideal, we can define linear maps
$$\chi_p :\T/\cN\to\T/\cN,\qquad\overline q\mapsto \overline{pq}$$
for $p\in\T$, which are bounded by \eqref{e:ineq}.

Now
\begin{equation}
\label{e:gns}
\langle\overline p,\overline q\rangle=L(pq^\star)
\end{equation}
is a scalar product on $\T/\cN$, and we let $H$ denote the completion of $\T/\cN$ with respect to this scalar product. Each $\chi_p$ extends to a bounded operator $\hat \chi_p$ on $H$, and the map
\begin{equation}\label{e:creation}
\pi: \T \to \cB(H), \qquad p\mapsto \hat \chi_p
\end{equation}
is clearly a $\star$-representation with $\ker\pi=\cN$. Let $\cF$ be the closure of $\pi(\T)$ in $\cB(H)$ with respect to the weak operator topology. The map
$$\tau:\pi(\T)\to \R,\qquad \hat \chi_p\mapsto L(p)$$
is a faithful tracial state on $\pi(\T)$ by $L(\cN)=\{0\}$ and \eqref{e:tr}. Since
$$\tau(\hat \chi_p)=\langle \overline p,\overline{1}\rangle,$$
$\tau$ extends uniquely to a faithful normal tracial state on $\cF$.

Next we claim that $\pi(\skinnyT)=\R$. Observe that $\overline{1}\in H$ is a cyclic vector for $\pi$ by construction and $L(p)=\langle \pi(p)\overline{1},\overline{1}\rangle$. Suppose $\pi(\skinnyT)\neq\R$. If $\cE$ denotes the weak closure of $\pi(\skinnyT)$ in $\cF$, then $\cE$ is a central von Neumann subalgebra of $\cF$; since $\cE\neq\R$ and all the elements of $\cE$ are self-adjoint, there is a nontrivial projection $P\in\cE$. Since $\overline{1}$ is cyclic for $\pi$, we have $P\overline{1}\neq0$ and $(1-P)\overline{1}\neq0$. Hence we can define linear functionals $L_i$ on $\T$ by
$$L_1(p)=\frac{\langle\pi(p)P\overline{1},P\overline{1}\rangle}{\|P\overline{1}\|^2}
\qquad\text{and}\qquad
L_2(p)=\frac{\langle\pi(p)(1-P)\overline{1},(1-P)\overline{1}\rangle}{\|(1-P)\overline{1}\|^2}$$
for all $p\in \T$. One easily checks that $L$ is a convex combination of $L_1$ and $L_2$, $L_i(1)=1$ and $L_i(\cQ)=\R_{\ge0}$. Furthermore, since $P$ is a weak limit of $\{\pi(s_n)\}_n$ for some $s_n\in\skinnyT$
and
$$\langle\pi(p-\Trace(p))\pi(s_n)\overline{1},\overline{1}\rangle
=L(s_n(p-\Trace(p)))=L(s_np-\Trace(s_np)))=0,$$
we also have
$$\langle\pi(p-\Trace(p))P\overline{1},P\overline{1}\rangle=
\langle\pi(p-\Trace(p))P\overline{1},\overline{1}\rangle=0$$
so $L_i(U)=\{0\}$. Therefore $L_i\in\cC$, so $L=L_1=L_2$ by the extreme property of $L$. Then for $\lambda=\|P\overline{1}\|^2$,
$$\langle \pi(p)\overline{1},\lambda\overline{1}\rangle
=\lambda\langle\pi(p)\overline{1},\overline{1}\rangle
=\langle\pi(p)P\overline{1},P\overline{1}\rangle
=\langle P\pi(p)\overline{1},P\overline{1}\rangle
=\langle \pi(p)\overline{1},P\overline{1}\rangle$$
for all $p\in \T$. Therefore $P\overline{1}=\lambda\overline{1}$ since $\overline{1}$ is a cyclic vector for $\pi$. So $\lambda\in\{0,1\}$ since $P$ is a projection, a contradiction.

Let $\uX:=(\hat \chi_{x_1},\dots,\hat \chi_{x_n})$. This is a tuple of self-adjoint operators in $\cF$, and $\pi(\skinnyT)=\R$ implies $p( \uX)=\hat \chi_p$ for all $p\in\T$. Therefore $\uX \in \mathcal{D}_\cQ^{\cF,\tau}$ by \eqref{e:gns} and $L(\cQ) \subseteq \R_{\ge0}$. Finally $a(\uX)=\tau(\hat \chi_a)=L(a)<0$.
\end{proof}

\end{appendices}


\begin{thebibliography}{PKRR{\etalchar{+}}19}

\bibitem[AL12]{anjos2011handbook}
Miguel~F. Anjos and Jean~B. Lasserre, editors.
\newblock {\em Handbook on semidefinite, conic and polynomial optimization},
  volume 166 of {\em International Series in Operations Research \& Management
  Science}.
\newblock Springer, New York, 2012.

\bibitem[ARU97]{ARU97}
Shavkat Ayupov, Abdugafur Rakhimov, and Shukhrat Usmanov.
\newblock {\em Jordan, real and {L}ie structures in operator algebras}, volume
  418 of {\em Mathematics and its Applications}.
\newblock Kluwer Academic Publishers Group, Dordrecht, 1997.

\bibitem[Bar02]{Bar}
Alexander Barvinok.
\newblock {\em A course in convexity}, volume~54 of {\em Graduate Studies in
  Mathematics}.
\newblock American Mathematical Society, Providence, RI, 2002.

\bibitem[BCKP13]{nctrace}
Sabine Burgdorf, Kristijan Cafuta, Igor Klep, and Janez Povh.
\newblock The tracial moment problem and trace-optimization of polynomials.
\newblock {\em Math. Program.}, 137(1-2, Ser. A):557--578, 2013.

\bibitem[BCS20]{bardet}
Ivan Bardet, Beno\^{\i}t Collins, and Gunjan Sapra.
\newblock Characterization of equivariant maps and application to entanglement
  detection.
\newblock {\em Ann. Henri Poincar\'{e}}, 21(10):3385--3406, 2020.

\bibitem[Bel64]{bell1964einstein}
John~S Bell.
\newblock On the {E}instein {P}odolsky {R}osen paradox.
\newblock {\em Physics Physique Fizika}, 1(3):195, 1964.

\bibitem[BKP16]{burgdorf16}
Sabine Burgdorf, Igor Klep, and Janez Povh.
\newblock {\em Optimization of polynomials in non-commuting variables}.
\newblock SpringerBriefs in Mathematics. Springer, [Cham], 2016.

\bibitem[BMV75]{bessis1975monotonic}
Daniel Bessis, Pierre Moussa, and Matteo Villani.
\newblock Monotonic converging variational approximations to the functional
  integrals in quantum statistical mechanics.
\newblock {\em J. Math. Phys.}, 16(11):2318--2325, 1975.

\bibitem[BRGP12]{BRGP}
Cyril Branciard, Denis Rosset, Nicolas Gisin, and Stefano Pironio.
\newblock Bilocal versus nonbilocal correlations in entanglement-swapping
  experiments.
\newblock {\em Phys. Rev. A}, 85:032119, Mar 2012.

\bibitem[CF98]{curto1998flat}
Ra\'{u}l~E. Curto and Lawrence~A. Fialkow.
\newblock Flat extensions of positive moment matrices: recursively generated
  relations.
\newblock {\em Mem. Amer. Math. Soc.}, 136(648):x+56, 1998.

\bibitem[Cha16]{Chaves}
Rafael Chaves.
\newblock Polynomial {B}ell inequalities.
\newblock {\em Phys. Rev. Lett.}, 116(1):010402, 6, 2016.

\bibitem[CHSH69]{clauser1969proposed}
John~F. Clauser, Michael~A. Horne, Abner Shimony, and Richard~A. Holt.
\newblock Proposed experiment to test local hidden-variable theories.
\newblock {\em Phys. rev. lett.}, 23(15):880, 1969.

\bibitem[CKP11]{cafuta2011ncsostools}
Kristijan Cafuta, Igor Klep, and Janez Povh.
\newblock N{CSOS}tools: a computer algebra system for symbolic and numerical
  computation with noncommutative polynomials.
\newblock {\em Optim. Methods Softw.}, 26(3):363--380, 2011.

\bibitem[CKP12]{cafuta2012constrained}
Kristijan Cafuta, Igor Klep, and Janez Povh.
\newblock Constrained polynomial optimization problems with noncommuting
  variables.
\newblock {\em SIAM J. Optim.}, 22(2):363--383, 2012.

\bibitem[DLTW08]{doherty2008quantum}
Andrew~C. Doherty, Yeong-Cherng Liang, Ben Toner, and Stephanie Wehner.
\newblock The quantum moment problem and bounds on entangled multi-prover
  games.
\newblock In {\em 2008 23rd Annual IEEE Conference on Computational
  Complexity}, pages 199--210. IEEE, 2008.

\bibitem[dOHMP09]{engineeringFRAG}
Mauricio~C. de~Oliveira, J.~William Helton, Scott~A. McCullough, and Mihai
  Putinar.
\newblock Engineering systems and free semi-algebraic geometry.
\newblock In {\em Emerging applications of algebraic geometry}, volume 149 of
  {\em IMA Vol. Math. Appl.}, pages 17--61. Springer, New York, 2009.

\bibitem[Dyk94]{Dyk94}
Kenneth~J. Dykema.
\newblock Factoriality and {C}onnes' invariant {$T({\mathcal M})$} for free
  products of von {N}eumann algebras.
\newblock {\em J. Reine Angew. Math.}, 450:159--180, 1994.

\bibitem[EW01]{eggeling}
Tilo Eggeling and Reinhard~F. Werner.
\newblock Separability properties of tripartite states with {$U\otimes U\otimes
  U$} symmetry.
\newblock {\em Phys. Rev. A (3)}, 63(4):042111, 15, 2001.

\bibitem[FN14]{fukuda2014asymptotically}
Motohisa Fukuda and Ion Nechita.
\newblock Asymptotically well-behaved input states do not violate additivity
  for conjugate pairs of random quantum channels.
\newblock {\em Comm. Math. Phys.}, 328(3):995--1021, 2014.

\bibitem[GdLL18]{Gribling18}
Sander Gribling, David de~Laat, and Monique Laurent.
\newblock Bounds on entanglement dimensions and quantum graph parameters via
  noncommutative polynomial optimization.
\newblock {\em Math. Program.}, 170(1, Ser. B):5--42, 2018.

\bibitem[GdLL19]{Gribling19}
Sander Gribling, David de~Laat, and Monique Laurent.
\newblock Lower bounds on matrix factorization ranks via noncommutative
  polynomial optimization.
\newblock {\em Found. Comput. Math.}, to appear 2019.

\bibitem[Had01]{Had}
Don Hadwin.
\newblock A noncommutative moment problem.
\newblock {\em Proc. Amer. Math. Soc.}, 129(6):1785--1791, 2001.

\bibitem[Hel02]{Helton02}
J.~William Helton.
\newblock ``{P}ositive'' noncommutative polynomials are sums of squares.
\newblock {\em Ann. of Math. (2)}, 156(2):675--694, 2002.

\bibitem[HHH01]{horodecki}
Michał Horodecki, Paweł Horodecki, and Ryszard Horodecki.
\newblock Separability of $n$-particle mixed states: necessary and sufficient
  conditions in terms of linear maps.
\newblock {\em Physics Letters A}, 283(1):1--7, 2001.

\bibitem[HKM11]{HKM11}
J.~William Helton, Igor Klep, and Scott McCullough.
\newblock Proper analytic free maps.
\newblock {\em J. Funct. Anal.}, 260(5):1476--1490, 2011.

\bibitem[HKT17]{GTineq2}
Fumio Hiai, Robert K\"{o}nig, and Marco Tomamichel.
\newblock Generalized log-majorization and multivariate trace inequalities.
\newblock {\em Ann. Henri Poincar\'{e}}, 18(7):2499--2521, 2017.

\bibitem[HM04]{Helton04}
J.~William Helton and Scott~A. McCullough.
\newblock A {P}ositivstellensatz for non-commutative polynomials.
\newblock {\em Trans. Amer. Math. Soc.}, 356(9):3721--3737, 2004.

\bibitem[Hub21]{huber2020positive}
Felix Huber.
\newblock Positive maps and trace polynomials from the symmetric group.
\newblock {\em J. Math. Phys}, 62(2):022203, 2021.

\bibitem[JNV{\etalchar{+}}20]{CECfalse}
Zhengfeng Ji, Anand Natarajan, Thomas Vidick, John Wright, and Henry Yuen.
\newblock {MIP* = RE}.
\newblock {\em arXiv preprint arXiv:2001.04383}, 2020.

\bibitem[KMP19]{klep2019sparse}
Igor Klep, Victor Magron, and Janez Povh.
\newblock Sparse noncommutative polynomial optimization.
\newblock {\em arXiv preprint arXiv:1909.00569}, 2019.

\bibitem[KPV18]{robust}
Igor Klep, Janez Povh, and Jurij Vol\v{c}i\v{c}.
\newblock Minimizer extraction in polynomial optimization is robust.
\newblock {\em SIAM J. Optim.}, 28(4):3177--3207, 2018.

\bibitem[KS08]{CECsohs}
Igor Klep and Markus Schweighofer.
\newblock Connes' embedding conjecture and sums of {H}ermitian squares.
\newblock {\em Adv. Math.}, 217(4):1816--1837, 2008.

\bibitem[K{\v{S}}17]{KS17}
Igor Klep and \v{S}pela {\v{S}}penko.
\newblock Free function theory through matrix invariants.
\newblock {\em Canad. J. Math.}, 69(2):408--433, 2017.

\bibitem[K{\v{S}}V18]{klep2018positive}
Igor Klep, \v{S}pela {\v{S}}penko, and Jurij Vol\v{c}i\v{c}.
\newblock Positive trace polynomials and the universal {P}rocesi--{S}chacher
  conjecture.
\newblock {\em Proc. Lond. Math. Soc.}, 117(6):1101--1134, 2018.

\bibitem[Lam13]{Lam13}
Tsit-Yuen Lam.
\newblock {\em A first course in noncommutative rings}, volume 131.
\newblock Springer Science \& Business Media, 2013.

\bibitem[Las01]{Las01sos}
Jean-Bernard Lasserre.
\newblock Global optimization with polynomials and the problem of moments.
\newblock {\em SIAM J. Optim.}, 11(3):796--817, 2000/01.

\bibitem[Lau09]{Laurent:Survey}
Monique Laurent.
\newblock Sums of squares, moment matrices and optimization over polynomials.
\newblock In {\em Emerging applications of algebraic geometry}, volume 149 of
  {\em IMA Vol. Math. Appl.}, pages 157--270. Springer, New York, 2009.

\bibitem[Lax58]{lax1957differential}
Peter~D. Lax.
\newblock Differential equations, difference equations and matrix theory.
\newblock {\em Comm. Pure Appl. Math.}, 11:175--194, 1958.

\bibitem[LLR08]{lasserre2008semidefinite}
Jean~Bernard Lasserre, Monique Laurent, and Philipp Rostalski.
\newblock Semidefinite characterization and computation of zero-dimensional
  real radical ideals.
\newblock {\em Foundations of Computational Mathematics}, 8(5):607--647, 2008.

\bibitem[LS04]{lieb2004equivalent}
Elliott~H. Lieb and Robert Seiringer.
\newblock Equivalent forms of the {B}essis-{M}oussa-{V}illani conjecture.
\newblock {\em J. Statist. Phys.}, 115(1-2):185--190, 2004.

\bibitem[Mar08]{marshallbook}
Murray Marshall.
\newblock {\em Positive polynomials and sums of squares}, volume 146 of {\em
  Mathematical Surveys and Monographs}.
\newblock American Mathematical Society, Providence, RI, 2008.

\bibitem[McC01]{McCullSOS}
Scott McCullough.
\newblock Factorization of operator-valued polynomials in several non-commuting
  variables.
\newblock {\em Linear Algebra Appl.}, 326(1-3):193--203, 2001.

\bibitem[MKKK10]{MKKK10}
Kazuo Murota, Yoshihiro Kanno, Masakazu Kojima, and Sadayoshi Kojima.
\newblock A numerical algorithm for block-diagonal decomposition of matrix
  {$*$}-algebras with application to semidefinite programming.
\newblock {\em Jpn. J. Ind. Appl. Math.}, 27(1):125--160, 2010.

\bibitem[NKI02]{NKI}
Koji Nagata, Masato Koashi, and Nobuyuki Imoto.
\newblock Configuration of separability and tests for multipartite entanglement
  in {B}ell-type experiments.
\newblock {\em Phys. Rev. Lett.}, 89(26):260401, 4, 2002.

\bibitem[NPA08]{navascues2008convergent}
Miguel Navascu{\'e}s, Stefano Pironio, and Antonio Ac{\'\i}n.
\newblock A convergent hierarchy of semidefinite programs characterizing the
  set of quantum correlations.
\newblock {\em New J. Phys.}, 10(7):073013, 2008.

\bibitem[NT14]{netzer2014hyperbolic}
Tim Netzer and Andreas Thom.
\newblock Hyperbolic polynomials and generalized {C}lifford algebras.
\newblock {\em Discrete Comput. Geom.}, 51(4):802--814, 2014.

\bibitem[PHBB17]{PHBB}
Victor Pozsgay, Flavien Hirsch, Cyril Branciard, and Nicolas Brunner.
\newblock Covariance {B}ell inequalities.
\newblock {\em Phys. Rev. A}, 96(6):062128, 13, 2017.

\bibitem[PKRR{\etalchar{+}}19]{pozas2019bounding}
Alejandro Pozas-Kerstjens, Rafael Rabelo, \L{}ukasz Rudnicki, Rafael Chaves,
  Daniel Cavalcanti, Miguel Navascu\'{e}s, and Antonio Ac\'{\i}n.
\newblock Bounding the sets of classical and quantum correlations in networks.
\newblock {\em Phys. Rev. Lett.}, 123(14):140503, 6, 2019.

\bibitem[PNA10]{pironio2010convergent}
Stefano Pironio, Miguel Navascu{\'e}s, and Antonio Ac{\'\i}n.
\newblock Convergent relaxations of polynomial optimization problems with
  noncommuting variables.
\newblock {\em SIAM J. Optim.}, 20(5):2157--2180, 2010.

\bibitem[Pro76]{P76}
Claudio Procesi.
\newblock The invariant theory of {$n\times n$} matrices.
\newblock {\em Adv. Math.}, 19(3):306--381, 1976.

\bibitem[PS98]{PS76}
George P\'{o}lya and Gabor Szeg\H{o}.
\newblock {\em Problems and theorems in analysis. {II}}.
\newblock Classics in Mathematics. Springer-Verlag, Berlin, 1998.
\newblock Theory of functions, zeros, polynomials, determinants, number theory,
  geometry, Translated from the German by C. E. Billigheimer, Reprint of the
  1976 English translation.

\bibitem[Put93]{Putinar1993positive}
Mihai Putinar.
\newblock Positive polynomials on compact semi-algebraic sets.
\newblock {\em Indiana Univ. Math. J.}, 42(3):969--984, 1993.

\bibitem[PV09]{pal2009quantum}
K\'{a}roly~F. P\'{a}l and Tam\'{a}s V\'{e}rtesi.
\newblock Quantum bounds on {B}ell inequalities.
\newblock {\em Phys. Rev. A (3)}, 79(2):022120, 12, 2009.

\bibitem[RTAL13]{Riener13Symmetries}
Cordian Riener, Thorsten Theobald, Lina~Jansson Andr\'{e}n, and Jean-Bernard
  Lasserre.
\newblock Exploiting symmetries in {SDP}-relaxations for polynomial
  optimization.
\newblock {\em Math. Oper. Res.}, 38(1):122--141, 2013.

\bibitem[SBT17]{GTineq1}
David Sutter, Mario Berta, and Marco Tomamichel.
\newblock Multivariate trace inequalities.
\newblock {\em Comm. Math. Phys.}, 352(1):37--58, 2017.

\bibitem[Sch91]{Schmudgen91sos}
Konrad Schmüdgen.
\newblock {The K-moment Problem for Compact Semi-algebraic Sets}.
\newblock {\em Mathematische Annalen}, 289(2):203--206, 1991.

\bibitem[SIG98]{skelton1997unified}
Robert~E. Skelton, Tetsuya Iwasaki, and Karolos~M. Grigoriadis.
\newblock {\em A unified algebraic approach to linear control design}.
\newblock The Taylor \& Francis Systems and Control Book Series. Taylor \&
  Francis, Ltd., London, 1998.

\bibitem[Sta13]{stahl2013proof}
Herbert~R. Stahl.
\newblock Proof of the {BMV} conjecture.
\newblock {\em Acta Math.}, 211(2):255--290, 2013.

\bibitem[Tak02]{Tak02}
Masamichi Takesaki.
\newblock {\em Theory of operator algebras. {I}}, volume 124 of {\em
  Encyclopaedia of Mathematical Sciences}.
\newblock Springer-Verlag, Berlin, 2002.
\newblock Reprint of the first (1979) edition, Operator Algebras and
  Non-commutative Geometry, 5.

\bibitem[Uff02]{Uffink}
Jos Uffink.
\newblock Quadratic {B}ell inequalities as tests for multipartite entanglement.
\newblock {\em Phys. Rev. Lett.}, 88(23):230406, 4, 2002.

\bibitem[Vid59]{vidav}
Ivan Vidav.
\newblock On some $*$-regular rings.
\newblock {\em Acad. Serbe Sci. Publ. Inst. Math.}, 13:73--80, 1959.

\bibitem[Wer89]{werner}
Reinhard~F. Werner.
\newblock Quantum states with einstein-podolsky-rosen correlations admitting a
  hidden-variable model.
\newblock {\em Phys. Rev. A}, 40:4277--4281, Oct 1989.

\bibitem[WSV12]{wolkowicz2012handbook}
Henry Wolkowicz, Romesh Saigal, and Lieven Vandenberghe.
\newblock {\em Handbook of semidefinite programming: theory, algorithms, and
  applications}, volume~27.
\newblock Springer Science \& Business Media, 2012.

\end{thebibliography}
\newcommand{\etalchar}[1]{$^{#1}$}

\end{document}